\documentclass{article}
\usepackage{fullpage,authblk}
\usepackage{microtype}%if unwanted, comment out or use option "draft"

\usepackage[linesnumbered,longend,noresetcount,boxed]{algorithm2e}
\usepackage{amsmath,tikz,graphicx,etoolbox,caption}
\usepackage{fancybox}
\usepackage{refcount}
\usepackage{subcaption,amsthm,url}

\newtoggle{abs}
\toggletrue{abs}
\togglefalse{abs}
\newtoggle{absfalse}
\iftoggle{abs}{\togglefalse{absfalse}}{\toggletrue{absfalse}}
\SetAlFnt{\small}

\newtheorem{claim}{Claim}
\newtheorem{observation}{Observation}

\newtheorem{theorem}{Theorem}
\newtheorem{definition}{Definition}
\newtheorem{lemma}{Lemma}
\newenvironment{proofof}[1]{\smallskip\par\noindent{\sl Proof of #1}:\enspace}{\hspace*{\fill} $\openbox$}
\newenvironment{proof_of_claim}{\noindent {\it Proof of Claim: }}{\hspace*{\fill} }

\title{A Duality Based 2-Approximation Algorithm for Maximum~Agreement Forest\iftoggle{abs}{ [Extended Abstract]}{}\footnote{\iftoggle{abs}{Full version is available at \url{http://arxiv.org/abs/1511.06000}.}{}
This work was initiated when the authors were visitors of Leen Stougie at the Tinbergen Institute. FS was supported in part by the Simon Prize for Excellence in the Teaching of Mathematics at William \& Mary. AvZ was supported in part by a William \& Mary Summer Research Award, NSF Prime Award: HRD-1107147, Women in Scientific Education (WISE) and by a grant from the Simons Foundation (\#359525, Anke Van Zuylen).}}
\author[1]{Frans Schalekamp}
\author[2]{Anke van Zuylen}
\author[3]{Suzanne van der Ster}
\affil[1]{School of Operations Research \& Information Engineering, Cornell University, Ithaca, NY, \texttt{fms9@cornell.edu}}
\affil[2]{Department of Mathematics, College of William \& Mary, Williamsburg, VA, \texttt{anke@wm.edu}}
\affil[3]{Institut f\"ur Informatik, Technische Universit\"at M\"unchen, Germany, \texttt{ster@in.tum.de} }

\DeclareMathOperator{\lca}{lca}

\DeclareMathOperator{\minimize}{minimize}

\newcommand{\resolvepair}{\text{\sc ResolvePair}}
\newcommand{\resolveset}{\text{\sc ResolveSet}}
\newcommand{\preprocess}{\text{\sc Preprocess}}
\usepackage{authblk}

\DontPrintSemicolon

\newcommand{\smallarbitrarytreeat}[1]{\draw #1 -- ++(-.5,-1) -- ++(1,0) -- ++(-.5,1);}
\newcommand{\drawnode}[1]{\draw [fill] #1 circle [radius=.07];}

\newcommand{\ot}{\ensuremath{\leftarrow}}
\newcommand{\dual}{\fbox}

\begin{document}
\maketitle

\setcounter{footnote}{0}
 \begin{abstract}
We give a 2-approximation algorithm for the Maximum Agreement Forest 
problem on two rooted binary trees. This NP-hard 
problem has been studied extensively in the past two decades, since it can be used to compute the Subtree Prune-and-Regraft (SPR) distance between two phylogenetic trees. Our result improves on the very recent 
2.5-approximation algorithm due to Shi, Feng, You and Wang (2015). Our 
algorithm is the first approximation algorithm for this problem that 
uses LP duality in its analysis.
\end{abstract}

\section{Introduction}
Evolutionary relationships are often modeled by a rooted tree, where the leaves are a set of species, and internal nodes are (putative) common ancestors of the leaves below the internal node. Such phylogenetic trees date back to Darwin~\cite{Darwin37}, who used them in his notebook to elucidate his thoughts on evolution.

The topology of phylogenetic trees can be based on different sources of data, e.g., morphological data, behavioral data, genetic data, etc., which can lead to different phylogenetic trees on the same set of species. Different measures have been proposed to measure the similarity of (or distance between) different phylogenetic trees on the same set of species (or individuals). 
Using the size of a maximum subtree common to both input trees as a similarity measure
was proposed by Gordon~\cite{Gordon79}. The problem of finding such a subtree is now known as the Maximum Agreement Subtree Problem, and has been studied extensively. Steel and Warnow~\cite{Steel93} are the first to give a polynomial-time algorithm for this problem. Their approach is refined to an $O(n^{1.5}\log n)$-time algorithm by Farach and Thorup~\cite{Farach94}, who subsequently show their algorithm is optimal, unless unweighted bipartite matching can be solved in near linear time~\cite{Farach97}.

There exist non-tree-like evolutionary processes that preclude the existence of a phylogenetic tree, so-called reticulation events (such as hybridization, recombination and horizontal gene transfer). In this context, a particularly meaningful measure of comparing phylogenetic trees is the SPR-distance measure (where SPR is short for Subtree Prune-and-Regraft): this measure provides a lower bound on a certain type of these non-tree evolutionary events. The problem of finding the exact value of this measure for a set of species motivated the formulation of the Maximum Agreement Forest Problem (MAF) by Hein, Jian, Wang and Zhang~\cite{Hein96}.  
Since the introduction by Hein et al., MAF has been extensively studied, including several variants, such as the problem where the input consists of more than two trees, whether the input trees are rooted or unrooted, binary or non-binary. In our paper, we concentrate on MAF on two rooted binary trees.
 
For ease of defining the solutions to MAF on two rooted binary trees, we think of the input trees as being directed, where all edges are directed away from the root. Given two rooted binary trees on the same leaf set ${\cal L}$, the MAF problem asks to find a minimum set of edges to 
delete from the two trees, so that the directed trees in the resulting two forests can be paired up into pairs of ``isomorphic'' trees. Two trees 
are isomorphic if they contain the same nodes from ${\cal L}$ and recursively removing nodes of out-degree zero that are not 
in ${\cal L}$ and contracting two arcs incident to a node of in-degree and out-degree one, results in the same binary tree. An alternative (but equivalent) definition will be given in Section~\ref{sec:prelim}.

The problem of finding a MAF on two rooted binary trees has been extensively studied, although unfortunately some of the published results later turned out to have subtle errors. First of all, Allen and Steel~\cite{Allen2001} point out that the claim by Hein et al. that solving MAF on two rooted directed trees computes the SPR-distance between the trees is incorrect. Bordewich and Semple~\cite{Bordewich04} show how to redefine MAF for rooted directed trees so that it is indeed the case that the optimal objective value of MAF is equal to the SPR-distance. In the paper in which they introduce MAF, Hein et al.~\cite{Hein96} proved NP-hardness and they give an approximation algorithm for the problem, the approximation guarantee of which turned out to be slightly higher than what was claimed. Bordewich and Semple~\cite{Bordewich04} show that, for their corrected definition of MAF, NP-hardness continues to hold. Other approximation algorithms followed~\cite{Rodrigues07,Bonet06,Bordewich08}. The current best approximation ratio for MAF is 2.5, due to Shi, Feng, You and Wang~\cite{Shietal15}, and Rodrigues~\cite{Rodrigues03} has shown that MAF is MAXSNP-hard. In addition, there is a body of work on other approaches, such as Fixed-Parameter Tractable (FPT) algorithms (e.g.,~\cite{Whidden09,Whidden13}) and Integer Programming ~\cite{Wu09,Wu10}.

In this paper we give an improved approximation algorithm for MAF with an analysis based on linear programming duality. 
Our 2-approximation algorithm differs from previous works in two aspects. First of all, in terms of bounding the optimal value, we construct a feasible dual solution, rather than arguing more locally about the objective of the optimal solution. Secondly, our algorithm itself also takes a more global approach, whereas the algorithms in previous works mainly consider local substructures of at most four leaves. In particular, we identify a minimal subtree\footnote{By subtree of a rooted tree, we mean a tree containing the leaves that are descendents of some particular node in the rooted tree (including this node itself), and all edges between them.}  of one of the two input trees for which the leaf set is incompatible, i.e., when deleting all other leaves from both trees, the remaining two trees are not isomorphic (such a minimal subtree can be found efficiently). We then use ``local'' operations which repeatedly look at two ``sibling'' leaves in the minimal incompatible subtree, and perform similar operations as those suggested by previous authors.

Preliminary tests were conducted using a proof-of-concept implementation of our algorithm in Java. Our preliminary results indicate that our algorithm finds a dual solution that in $44\%$ of the 1000 runs is equal to the optimal dual solution, and in $37\%$ of the runs is $1$ less than the optimal solution. The observed average approximation ratio is about $1.92$; following our algorithm with a simple greedy search algorithm decreases this to less than $1.28$. 

The remainder of our paper is organized as follows. In Section~\ref{sec:prelim}, we formally define the problem. In Section~\ref{sec:3approx}, we give a 3-approximation algorithm for MAF that introduces the dual linear program that is used throughout the remainder of the paper and gives a flavor of the arguments used to prove the approximation ratio of two. In Section~\ref{sec:key}, we give an overview of the 2-approximation algorithm and the key ideas in its analysis. In Section~\ref{sec:impl} we give more details on the randomly generated instances that we used to obtain our preliminary experimental results, and we conclude in Section~\ref{sec:concl} with some directions for future research.

\iftoggle{abs}{The full version of this paper~\cite{full} contains further details on the algorithm and a complete analysis that shows that its approximation ratio is 2.}

\section{Preliminaries}\label{sec:prelim}

The input to the Maximum Agreement Forest problem (MAF) consists of two rooted binary trees $T_1$ and $T_2$, where the leaves in each tree are labeled with the same label set ${\cal L}$. Each leaf has exactly one label, and each label in ${\cal L}$ is assigned to exactly one leaf in $T_1$, and one leaf in $T_2$. For ease of exposition, we sometimes think of the edges in the trees as being directed, so that there is a directed path from the root to each of the leaves.

We call the non-leaf nodes the internal nodes of the trees, and we let $V$ denote the set of all nodes (internal nodes and leaves) in $T_1\cup T_2$.
Given a tree containing $u$ and $v$, we let $\lca(u,v)$ denote the lowest (closest to the leaves, furthest from the root) common ancestor of $u$ and $v$. We let $\lca_1(u,v)$ and $\lca_2(u,v)$ denote $\lca(u,v)$ in tree $T_1$, respectively, $T_2$. We extend this notation to $\lca(U)$ which will denote the lowest common ancestor of a set of leaves $U$.
For three leaves $u,v,w$ and a rooted tree $T$, we use the notation $uv|w$ in $ T$ to denote that $\lca(u,v)$ is a descendent of $\lca(\{u,v,w\})$. 
A triplet $\{u,v,w\}$ of labeled leaves is {\em consistent} if $uv|w$ in $ T_1\Leftrightarrow uv|w$ in $ T_2$. The triplet is called {\em inconsistent} otherwise.
We call a set of leaves $L\subseteq {\cal L}$ a {\em compatible set}, if it does not contain an inconsistent triplet.

For a compatible set $L\subseteq {\cal L}$, define $V[L] := \{ v\in V : $ there exists a pair of leaves $u, u'$ in $L$ so that $v$ is on the path between $u$ and $u'$ in $T_1$ or $T_2\}$. Then, a partitioning $L_1, L_2, \ldots, L_p$ of ${\cal L}$ corresponds to a feasible solution to MAF with objective value $p-1$, if the sets $L_1, L_2, \ldots, L_p$ are compatible, and the sets $V[L_j]$ for $j=1, \ldots, p$ are node disjoint. 
Using this definition, we can write the following Integer Linear Program\footnote{This ILP was obtained in discussions with Neil Olver and Leen Stougie.} for MAF:
Let ${\cal C}$ be the collection of all compatible sets of leaves, and introduce a binary variable $x_L$ for every compatible set $L\in {\cal C}$, where the variable takes value 1 if the optimal solution to MAF has a tree with leaf set $L$. The constraints ensure that each leaf $v\in {\cal L}$ is in some tree in the optimal forest, and each internal node $v\in V\setminus{\cal L}$ is in at most one tree in the optimal forest. 
The objective encodes the fact that we need to delete $\sum_{L\in {\cal C}} x_L-1$ edges from each of $T_1$ and $T_2$ to obtain forests with $\sum_{L\in {\cal C}}x_L$ trees.

\[\begin{array}{lll}
\minimize &\sum_{L\in {\cal C}} x_L - 1,\\
\mbox{s.t.}& \sum_{L:v\in L} x_L = 1 & \forall v \in {\cal L},\\
&\sum_{L:v\in V[L]} x_L \le 1 & \forall v \in V\setminus {\cal L},\\
&x_L \in \{ 0, 1 \} &\forall L\in {\cal C}.
\end{array}\]

\paragraph*{Remark} 
The definition of MAF we use is {\it not} the definition that is now standard in the literature, but any (approximation) algorithm for our version can be used to get the same (approximation) result for the standard formulation: The standard formulation was introduced by Bordewich and Semple~\cite{Bordewich04} in order to ensure that the objective value of MAF is equal to the rooted SPR distance. They note that for this to hold, we need the additional requirement that the two forests must also agree on the tree containing the original root; in other words, the original roots of $T_1$ and $T_2$ should be contained in a tree with the same (compatible) subset of leaves. An easy reduction shows that we can solve this problem using our definition of MAF: given two rooted binary trees for which we want to compute the SPR distance, we can simply add one new label $\rho$, and for each of the two input trees, we add a new root which has an edge to the original root and an edge to a new node with label $\rho$.\footnote{This is essentially the formulation that is common in the literature, except that in order to ensure that {\it only leaves have labels}, we give the label $\rho$ to a new leaf that is an immediate descendent of the new root in both trees, rather than to the new root itself.} A solution to ``our'' MAF problem on this modified input defines a solution to Bordewich and Semple's problem on the original input with the same objective value and vice versa.

\section{Duality Based 3-Approximation Algorithm}\label{sec:3approx}

\subsection{Algorithm}
The algorithm we describe in this section is a variant of the algorithm of Rodrigues et al.~\cite{Rodrigues07} (see also Whidden and Zeh~\cite{Whidden09}).
The algorithm maintains two forests, $T'_1$ and $T'_2$ on the same leaf set ${\cal L}'$, and iteratively deletes edges from these forests. 
At the start, $T'_1$ is set equal to $T_1$, $T'_2$ to $T_2$ and ${\cal L}'$ to ${\cal L}$. 
The leaves in ${\cal L'}$ are called the {\it active} leaves. The algorithm will ensure that the leaves that are not active, will have been resolved in one of the two following ways: (1) they are part of a tree that contains only inactive leaves in both $T'_1$ and $T'_2$; these two trees then have the same leaf set, which is compatible, and they will be part of the final solution; or (2) an inactive leaf is {\em merged} with another leaf which is active, and in the final solution this inactive leaf will be in the same tree as the leaf it was merged with.

A tree is called active if it contains a leaf in ${\cal L'}$, and the tree is called inactive otherwise. 
An invariant of the algorithm is that there is a single active tree in $T_1'$.

We define the parent of a set of active leaves $W$ in a tree of a forest, denoted by $p( W )$, as the lowest node in the tree that is a common ancestor of $W$ and at least one other active node. (That is, $p( W )$ is undefined if there are no other active leaves in the tree that contains the leaves in $W$.) Note that the parent of a node is defined with respect to the current state of the algorithm, and not with respect to the initial input tree. 
If $W = \{ u \}$ is a singleton, we will also use the notation $p( u ) = p( \{ u \} )$. 
For a given tree or forest $T_i'$, for $i\in \{1,2\}$, we use the notation $p_i( W )$ to denote $p( W )$ in $T'_i$.

The operation in the algorithm that deletes edges from forest $T'_i$ is {\em cut off a subset of leaves} $W$ in $T'_i$. The edge that is deleted by this operation is the edge directly below $p_i( W )$ towards $W$ (provided $p_i( W )$ is defined). Note that this means that the algorithm maintains the property that each internal node has a path to at least one leaf in ${\cal L}$. 
This ensures that the number of trees with leaves in ${\cal L}$ in $T'_i$ is equal to the number of edges cut in $T'_i$ plus 1. It also ensures that the only leaves (nodes with outdegree 0) are the nodes in ${\cal L}$.

We will call two leaves $u$ and $v$ {\em a sibling pair} or {\em siblings} in a forest, if they belong to the same tree in the forest, and they are the only two leaves in the subtree rooted at the lowest common ancestor $\lca( u, v )$. Similarly, $u$ and $v$ are an {\it active} sibling pair in a forest, if they belong to the same tree in the forest, and  are the only two {\it active} leaves in the subtree rooted at the lowest common ancestor $\lca( u, v )$ (an equivalent definition is that $p( u ) = p( v )$ in the forest).

If leaves $u$ and $v$ are an active sibling pair in both $T'_1$ and $T'_2$, we {\em merge} one of the leaves (say $u$) with the other ($v$). This means that from that point on $v$ represents the subtree containing both $u$ and $v$, instead of just the leaf $v$ itself. This is accomplished by just making $u$ inactive.
Note that this merge operation can be performed recursively, where one or both of the leaves involved in the operation can already be leaves that represent subtrees. It is not hard to see that the subtree that is represented by an active leaf $v$ is one of the two subtrees rooted at the child of $p(v)$, namely the subtree that contains $v$.

If leaves $u$ and $v$ are not active siblings in $T'_2$ (and they are active siblings in $T_1'$), we can choose to {\em cut off an active subtree} between leaves $u$ and $v$. To describe this operation, let $W_1, W_2, \ldots, W_k$ be the active leaves of the active trees that would be created by deleting the path between $u$ and $v$ (both the nodes and the edges) in $T_2'$. Note that $p_2( W_\ell )$ is on the path between $u$ and $v$ for all $\ell\in\{1,2,\ldots,k\}$, because $u$ and $v$ are not active siblings. Cutting off an active subtree between leaves $u$ and $v$ now means cutting off any such a set $W_\ell$.

The algorithm is given in Algorithm~\ref{fig:3approx}. The boxed expressions refer 
to updates of the dual solution which will be discussed in Section~\ref{sec:3approxanalysis}. These expressions are only necessary for the analysis of the algorithm. 

\begin{theorem}\label{thm:3approx}
Algorithm~\ref{fig:3approx} is a 3-approximation algorithm for the Maximum Agreement Forest problem.
\end{theorem}
The proof of this theorem is given in the next subsection. It is clear the algorithm can be implemented to run in polynomial time. In Section~\ref{sec:3correct}, we show that the algorithm returns an agreement forest and we show that the number of edges deleted from $T_2$ by the algorithm can be upper bounded by three times the objective value of a feasible solution to the dual of a linear programming (LP) relaxation of MAF.  

\begin{algorithm}[ht]
\dual{$y_v \ot 0$ for all $v \in V$.} 
\;

\While{ there exist at least 2 active leaves}
{
	Find an active sibling pair $u,v$ in the active tree in $T_1'$.\;
	
	\eIf{ $u$ or $v$ is the only active leaf in its tree in $T_2'$ }
	{
		Cut off this node (say $u$) in $T_1'$ as well and make it inactive.  \dual{$y_u \ot 1$.} \;\label{alg:cut1}
	}	
	{
		\eIf{ $u$ and $v$ are active siblings in $T_2'$ }
		{
		Merge $u$ and $v$ (i.e., make $u$ inactive to ``merge'' it with $v$). \;\label{alg:merge}
		}
		{	
			\If{ $u$ and $v$ are in the same tree in $T_2'$, and this tree contains an active leaf $w$ such that $uv|w $ in $T_2$}
			{
				Cut off an active subtree $W$ between $u$ and $v$ in $T'_2$. \dual{Decrease $y_{p_2(W)}$ by $1$.} \;\label{alg:cutW}
			}
			Cut off $u$ and cut off $v$	in $T'_1$ and $T'_2$ and make them inactive. \dual{$y_u \ot 1$, $y_v \ot 1$, decrease $y_{\lca_1(u,v)}$ by 1.}  \;\label{alg:cut2}
		}
	}
}
\;
Make the remaining leaf (say $v$) inactive. \dual{$y_v \ot 1$.} \;\label{alg:cut3}

\caption{A 3-Approximation for Maximum Agreement Forest. 
The boxed text 
refers to updates of the dual solution as discussed in Section~\ref{sec:3approxanalysis}. } \label{fig:3approx}
\end{algorithm}
\setcounter{AlgoLine}{0}
\RestyleAlgo{boxed}
\subsection{Analysis of the 3-Approximation Algorithm}

\subsubsection{Correctness}\label{sec:3correct}
We need to show that the algorithm outputs an agreement forest. The trees of $T_1'$ and $T_2'$ each give a partitioning of ${\cal L}$, and clearly any internal node $v$ belongs to $V[L]$ for at most one set in the partitioning. It remains to show that the two forests give the {\em same} partitioning of ${\cal L}$ and that each set in the partitioning is compatible. 

The algorithm ends with all trees in $T_1'$ and $T_2'$ being inactive, and the algorithm maintains that the set of leaves represented by an active leaf $u$ (i.e., the leaves that were merged with $u$ (recursively), and $u$ itself) form the leaf set of a subtree in both $T'_1$ and $T'_2$. To be precise, it is the subtree rooted at one of the children of $p(u)$, namely the subtree that contains $u$. Furthermore note that this leaf set is compatible. This is easily verified by induction on the number of merges.

When $u$ is the only active leaf in its tree in both forests, then the trees containing $u$ in the two forests are thus guaranteed to have the same, compatible, set of leaves. Now, an inactive tree is created exactly when both $T_1'$ and $T_2'$ have an active tree in which some $u$ is the only active leaf (lines~\ref{alg:cut1}, \ref{alg:cut2} and \ref{alg:cut3}), and thus the two forests indeed induce the same partition of ${\cal L}$ into compatible sets.

\subsubsection{Approximation Ratio} \label{sec:3approxanalysis}
In order to prove the claimed approximation ratio, we will construct a feasible dual solution to the dual of the relaxation of the ILP given in Section~\ref{sec:prelim}. The dual LP is given in Figure~1(a). The dual LP has an optimal solution in which $0\le y_v\le 1$ for all $v\in {\cal L}$. The fact that $\{v\}$ is a compatible set implies that $y_v\le 1$ must hold for every $v\in {\cal L}$. Furthermore, note that changing the equality constraints of the primal LP to $\ge$-inequalities does not change the optimal value, and hence we may assume $y_v\ge 0$ for $v\in {\cal L}$. 

It will be convenient for our analysis to rewrite this dual by introducing additional variables for every (not necessarily compatible) set of labeled leaves. 
We will adopt the convention to use the letter $A$ to denote a set of leaves that is not necessarily compatible, and the letter $L$ to denote a set of leaves that is compatible (i.e., $L\in {\cal C}$). 
The dual LP can then be written as in Figure~1(b). Any solution to this new LP  can be transformed into a solution to the original dual LP by, for each $A$ such that $z_A> 0$, taking some leaf $v\in A$ and setting $y_v=y_v+z_A$ and $z_A=0$. This is feasible because the left-hand side of the first family of inequalities will not increase for any compatible set $L$, and it will decrease for $L$ such that $A \cap L \neq \emptyset$ and $v \not\in L$. Conversely, a solution to the original dual LP is feasible for this new LP by setting $z_A=0$, for every set of labeled leaves $A$.

\begin{figure}
\begin{subfigure}[t]{.396\textwidth}
\parbox{.395\textwidth}{\small
\[\begin{array}{lll}
\max &\sum_v y_v -1,\\
\mbox{s.t. }&\sum_{v\in V[L]} y_v \le 1 &\forall L\in {\cal C},\\
&y_v \le 0&\forall v\in V\setminus {\cal L}.\\
& \\
&
\end{array}\]
}

\caption{Dual LP}
\end{subfigure}
\quad
\begin{subfigure}[t]{.571\textwidth}
{\parbox{.57\textwidth}{\small
\[\begin{array}{lll}
\max &\sum_v y_v+\sum_{A\subseteq {\cal L}}  z_A - 1,\\
\mbox{s.t. }&\sum_{v\in V[L]} y_v +\sum_{A:A\cap L\neq \emptyset}z_A \le 1 & \forall L\in {\cal C},\\
&y_v \le 0& \forall v\in V\setminus {\cal L},\\
& y_v \ge 0& \forall v\in {\cal L},\\
&z_A \ge 0 &\forall A\subseteq{\cal L}.
\end{array}\]
}}
\caption{Reformulated dual LP}
\end{subfigure}
\caption{The dual of the LP relaxation for the ILP given in Section~\ref{sec:prelim}. The reformulated dual LP will be referred to as (D).}\label{fig:dual}
\end{figure}

We will refer to the left-hand side of the first family of constraints, i.e., $\sum_{v\in V[L]} y_v +\sum_{A:A\cap L\neq \emptyset}z_A$, as the {\em load} on set $L$.
\begin{definition}
The dual solution associated with a forest $T_2'$, obtained from $T_2$ by edge deletions, active leaf set ${\cal L}'$, and variables $y=\{y_v\}_{v\in V}$ is defined as $(y,z)$ where $z_A = 1$ exactly when $A$ is the active leaf set of a tree in $T'_2$, and $0$ otherwise. 
\iftoggle{abs}{}{The objective value $\sum_v y_v+\sum_{A\subseteq {\cal L}}  z_A - 1$ of the solution $(y,z)$ associated with forest $T_2'$, active leaf set ${\cal L}'$, and variables $y$ will be denoted by $D(T_2',{\cal L}',y)$.}
\end{definition}
We will sometimes use ``the dual solution'' to refer to the dual solution associated with $T_2',{\cal L}'$ and $y$ when the forest, active leaf set and $y$-values are clear from the context. 

\begin{lemma}\label{lemma:3dual}
After every execution of the while-loop of Algorithm~\ref{fig:3approx}, 
the dual solution associated with $T_2'$, ${\cal L}'$, and $y$ is a feasible solution to (D) and the objective value\iftoggle{abs}{}{ $D(T_2',{\cal L}',y)$} of this solution is at least $\tfrac 13|E(T_2) \setminus E(T'_2)|$.  
\end{lemma}

\begin{proof}
We prove the lemma by induction on the number of iterations. Initially, $z_{\cal L} = 1$, and all other variables are equal to $0$. Clearly, this is a feasible solution with objective value $0$. 

Observe that the dual solution maintained by the algorithm satisfies that $y_u=0$ while $u$ is active. Therefore, if there is a single active leaf $u$ in a tree in $T_2'$, then making this leaf $u$ inactive and setting $y_u=1$ does not affect dual feasibility and the dual objective value, since making $u$ inactive decreases $z_{\{u\}}$ from 1 to 0. Also note that merging two active leaves (and thus making one of the two leaves inactive), replaces the set of active leaves $A$ in an active tree in $T_2'$ by a smaller set $A'\subset A$, with $A'\neq \emptyset$. Hence, the dual solution changes from having $z_A=1$ to having $z_{A'}=1$, which clearly does not affect dual feasibility or the dual objective value.
Hence, we only need to verify that the dual solution remains feasible and its objective increases sufficiently for operations of the algorithm that cut edges from $T_2'$, i.e., lines~\ref{alg:cutW} and~\ref{alg:cut2}. 

In line~\ref{alg:cutW}, one edge is cut in $T_2'$, $y_{p_2(W)}$ decreases by $1$. Let $A$ be the set of active leaves in the tree containing $W$ in $T_2'$ before cutting off $W$. $z_A$ decreases by $1$, $z_{ A \setminus W }$  increases by $1$, $z_W$ increases by $1$.   The only sets $L$ for which the left-hand side potentially increases are sets $L$ so that $W \cap L \neq \emptyset$ and $(A \setminus W) \cap L \neq \emptyset$. However, $p_2( W ) \in V[L]$ for such sets $L$, and since $y_{p_2(W)}$ is decreased by $1$, the load is not increased for any compatible set $L$.
The dual objective is unchanged, but will change in line~\ref{alg:cut2} of the algorithm, as we will show next.

In line~\ref{alg:cut2}, let $A_u$ be the set of active leaves in the tree in $T'_2$ containing $u$ at the start of line~\ref{alg:cut2} in the algorithm, and  $A_v$ be the set of active leaves in the tree in $T'_2$ containing $v$. Note that $A_u \setminus \{ u, v \} \neq \emptyset$: if $v\not\in A_u$, then this holds because otherwise we would execute line~\ref{alg:cut1}, and if $v\in A_u$, then this holds because $u,v$ are not active siblings at the start of line~\ref{alg:cutW}, and if $u,v$ became active siblings after executing line~\ref{alg:cutW}, then the condition for line~\ref{alg:cutW} implies that there exists $w\in A_u$ such that $uv|w$ in $ T_2$.

The fact that  $A_u \setminus \{ u, v \} \neq \emptyset$ (and, by symmetry $A_v \setminus \{ u, v \} \neq \emptyset$) implies that the  total value of $\sum_A z_A+y_u+y_v$ increases by 2. Since we also decrease $y_{\lca_1( u, v )}$ by $1$ the total increase in the objective of the dual solution by line~\ref{alg:cut2} is $1$. Also, in lines~\ref{alg:cutW} and \ref{alg:cut2}, a total of at most three edges are cut in $T_2'$.

It remains to show that executing line~\ref{alg:cut2} does not make the dual solution $(y, z)$ infeasible. 
Note that for each active set $A'$ with $z_{A'} = 1$ before cutting off $u$ and $v$, there is exactly one unique active subset $A \subseteq A'$ with $z_{A}=1$ after cutting off $u$ and $v$. 
Therefore the total value of $\sum_{ A: A\cap L \neq \emptyset } z_A$ does not increase after cutting off $u$ and $v$ for any $L \subseteq {\cal L}$.

For any $L$ that includes one of $u, v$, and at least one other active leaf, it must be the case that $\lca_1( u, v ) \in V[ L ]$, because all active leaves are in one tree in $T_1'$, and $u$ and $v$ were active siblings in $T_1'$ at the start. Hence the only compatible sets $L$ for which the load on $L$ potentially increases by $1$ because of an increase in $\sum_{x\in L}y_x$ are sets $L$ that include both $u$ and $v$. 
We discern two cases.

\noindent
Case 1: An active subtree $W$ was cut off in line~\ref{alg:cutW}.  In this case, the load on $L$ was decreased by $1$ in line~\ref{alg:cutW}, compensating for the increase in line~\ref{alg:cut2}: $V[L]$ contains all nodes on the path between $u$ and $v$ in $T_2$, and hence also $p_2( W )$. It cannot contain a leaf $x \in W$, because $\{ u, v, x \}$ form an inconsistent triplet (because $uv|x$ in $T_1$).

\noindent
Case 2: No active  subtree $W$ was cut off in line~\ref{alg:cutW}.  In this case, the value of $\sum_{ A: A\cap L \neq \emptyset } z_A$ is decreased by at least $1$: If $u$ and $v$ are in the same tree in $T_2'$ before cutting off $u$ and $v$, then this tree contains no leaves $x$ such that $uv|x$ in $ T_2$ since otherwise an active subtree $W$ would have been cut off. Hence, $L$ does not contain any active leaf $x$ in the active tree that remains after cutting off $u$ and $v$ in $T_2'$, since any such leaf $x$ does not have $uv|x$ in $ T_2$ and therefore forms an inconsistent triplet with $u$ and $v$. Since $L$ does contain active leaves in the tree containing $u$ and $v$ in $T_2'$ before cutting off $u$ and $v$ (namely, $u$ and $v$ themselves), the value of $\sum_{ A: A\cap L \neq \emptyset } z_A$ indeed decreases by $1$.

If $u$ and $v$ are not in the same tree in $T_2'$ before cutting off $u$ and $v$, then a similar argument holds. Since $T_2'$ is obtained from $T_2$ by deleting edges, at least one of the two active trees containing $u$ and $v$ contains no leaves $x$ such that $uv|x$ in $ T_2$. Without loss of generality, suppose that this holds for the tree containing $u$. Then, $L$ does not contain any active leaves in the active tree remaining after $u$ is cut off, and hence $\sum_{ A: A\cap L \neq \emptyset } z_A$ decreases by at least $1$.
\end{proof}
By weak duality, we have that the objective value of any feasible solution to (D) provides a lower bound on the objective value of any feasible solution to the LP relaxation of our ILP for MAF, and hence also on the optimal value of the ILP itself. Theorem~\ref{thm:3approx} thus follows from Lemma~\ref{lemma:3dual} and the correctness shown in Section~\ref{sec:3correct}.

\section{Overview of the 2-Approximation Algorithm}\label{sec:key}
In this section, we begin by giving an outline of the key ideas of our 2-approximation algorithm. We then give an overview of the complete algorithm that we call the ``Red-Blue Algorithm''. 

One of the main ideas behind our 2-approximation is the consideration of the following two ``essential'' cases.
The first ``essential'' case is the case where we have an active sibling pair $u,v$ in $T_1'$ that are (i) active siblings in $T_2'$, or (ii) in different trees in $T_2'$, or (iii) the tree in $T_2'$ containing $u,v$ does {\it not} contain an active leaf $w$ such that $uv|w$ in $ T_2$. Then, it is easy to verify, using the arguments in the proof of Lemma~\ref{lemma:3dual}, that Algorithm~\ref{fig:3approx} ``works'': it increases the dual objective value by at least half of the increase in $|E(T_2)\setminus E(T_2')|$. We will say such a sibling pair $u,v$ is a ``success''.

The second ``essential'' case is the case where, in our current forest $T_1'$, there is a subtree containing exactly three active leaves, say $u,v,w$, where $uv|w$ in $ T_1$, and $\{u,v,w\}$ is an inconsistent triplet; assume without loss of generality that $uw|v$ in $ T_2$, and that the first ``essential'' case does not apply; in particular, this implies that $u,v$ are in the same tree in $T_2'$. 
It turns out that such an inconsistent triplet can be ``processed'' in a way that allows us to increase the objective value of the dual solution in such a way that it ``pays for'' half the increase in the number of edges cut from $T_2'$. There are a number of different cases to consider depending on whether all three leaves $u,v,w$ are in the same tree in $T_2'$ (case I) or not (case II), and whether the tree in $T_2'$ containing $w$ has an active leaf $x$ such that $xw|u$ in $ T_2'$ (case (a)) or not (case (b)). 

Figure~\ref{fig:triplet} gives an illustration of $T_1'$ and some possible configurations for $T_2'$. 
Consider for example case (I)(b). Since $\{u,v,w\}$ is an inconsistent triplet, it is not hard to see that any solution to MAF either has $v$ as a singleton component, or either $u$ or $w$ must be a singleton component. Indeed, we can increase the dual objective by 1, by updating the $y$-values as indicated by the circled values in Figure~\ref{fig:triplet} (a) and (b). The bold diagonal lines denote two edges that are cut (deleted from $T_2'$). Similar  arguments can be made for the other cases.
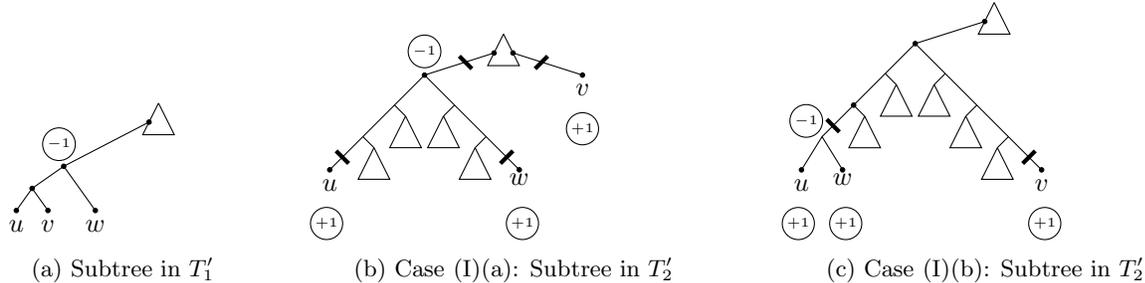
\begin{figure}
%\vspace{-2\baselineskip}
\begin{subfigure}[t]{.2\textwidth}
\begin{tikzpicture}[scale=.42]
\drawnode{(-3.5,-3.7)}
\draw (-4,-4.4) -- (-3.5,-3.7) -- (-3,-4.4);
\drawnode{(-4,-4.4)}
\drawnode{(-3,-4.4)}
\node [below] at (-4,-4.4) { $u$ };
\node [below] at (-3,-4.4) { $v$ };

\drawnode{(-1.5,-4.4)}
\node [below] at (-1.5,-4.4) { $w$ };

\draw (-3.5,-3.7) -- (-2.5,-3) -- (-1.5,-4.4);
\drawnode{(-2.5,-3)}

\draw (-2.5,-3) -- ++(2.7,1.4);
\drawnode{(.2,-1.6)}
\smallarbitrarytreeat{(.5,-1)}

\node [circle, draw, inner sep=1pt] at (-2.65,-2.3) {\tiny $-1$};
\end{tikzpicture}
\caption{Subtree in $T_1'$}
\end{subfigure}
\qquad
\begin{subfigure}[t]{.33\textwidth}
\begin{tikzpicture}[scale=.42]
\draw (-3,-3) -- (0,0) -- (3,-3);
\drawnode{(-3,-3)}
\node [below] at (-3,-3) {$u$};
\drawnode{(3,-3)}
\node [below] at (3,-2.83) {$w$};
\drawnode{(0,0)}

\smallarbitrarytreeat{(-.6,-1.3)}
\smallarbitrarytreeat{(-1.6,-2.3)}
\draw (-1.6,-2.3) -- (-1.95,-1.95);
\draw (-.6,-1.3) -- (-.95,-.95);

\smallarbitrarytreeat{(.6,-1.3)}
\smallarbitrarytreeat{(1.6,-2.3)}
\draw (1.6,-2.3) -- (1.95,-1.95);
\draw (.6,-1.3) -- (.95,-.95);

\smallarbitrarytreeat{(2.5,1.3)}
\draw (0,0) -- (2.2,0.7);
\drawnode{(2.2,0.7)}
\drawnode{(2.8,.7)}
\draw (5,0) -- (2.8,.7);
\drawnode{(5,0)}
\node [below] at (5,0) {$v$};

\draw [ultra thick] (-2.4,-2.8) -- (-2.8,-2.4); 
\draw [ultra thick] (2.4,-2.8) -- (2.8,-2.4); 
\draw [ultra thick] (1.5,.2) -- (1.1,.6); 
\draw [ultra thick] (3.5,.2) -- (3.9,.6); 

\node [circle, draw, inner sep=1pt] at (-3.1,-4.7) {\tiny$+1$};
\node [circle, draw, inner sep=1pt] at (3.1,-4.7) {\tiny$+1$};
\node [circle, draw, inner sep=1pt] at (0,.75) {\tiny $-1$};
\node [circle, draw, inner sep=1pt] at (5,-1.7) {\tiny$+1$};
\end{tikzpicture}
\caption{Case (I)(a): Subtree in $T_2'$}
\end{subfigure}
\qquad
\begin{subfigure}[t]{.33\textwidth}
\begin{tikzpicture}[scale=.42]
\draw (-2.95,-2.95) -- (0,0) -- (4,-4);
\drawnode{(-3.6,-4)}
\node [below] at (-3.6,-4) {$u$};
\drawnode{(4,-4)}
\node [below] at (4,-4) {$v$};
\drawnode{(0,0)}
\drawnode{(-2.3,-4)}
\node [below] at (-2.3,-3.83) {$w$};
\draw (-3.6,-4) -- (-2.95,-2.95) -- (-2.3,-4);
\smallarbitrarytreeat{(-.6,-1.3)}
\smallarbitrarytreeat{(-1.6,-2.3)}
\draw (-1.6,-2.3) -- (-1.95,-1.95);
\draw (-.6,-1.3) -- (-.95,-.95);
\smallarbitrarytreeat{(.6,-1.3)}
\smallarbitrarytreeat{(1.6,-2.3)}
\draw (1.6,-2.3) -- (1.95,-1.95);
\draw (.6,-1.3) -- (.95,-.95);
\smallarbitrarytreeat{(2.6,-3.3)}
\draw (2.6,-3.3) -- (2.95,-2.95);
\smallarbitrarytreeat{(2.5,1.3)}
\draw (0,0) -- (2.2,0.7);
\drawnode{(2.2,0.7)}
\draw [ultra thick] (3.4,-3.8) -- (3.8,-3.4); 
\draw [ultra thick] (-2.4,-2.8) -- (-2.8,-2.4); 
\node [circle, draw, inner sep=1pt] at (-3.7,-5.7) {\tiny$+1$};
\node [circle, draw, inner sep=1pt] at (-2.2,-5.7) {\tiny$+1$};
\node [circle, draw, inner sep=1pt] at (4.1,-5.7) {\tiny$+1$};
\node [circle, draw, inner sep=1pt] at (-3.45,-2.45) {\tiny$-1$};
\drawnode{(-1.95,-1.95)}
\end{tikzpicture}
\caption{Case (I)(b): Subtree in $T_2'$}
\end{subfigure}
\caption{Dealing with an inconsistent triplet: Circled values denote the $y$-variables that are set by the algorithm, and bold diagonal lines denote edges that are cut (deleted from $T_2'$). Triangles denote subtrees with active leaves (that may be empty). Note that there is a distinction between edges that are incident to the root of a subtree represented by a triangle, and edges that are incident to some internal node of the subtree. The latter edges are connected to a dot on the triangle.}\label{fig:triplet}
\end{figure}

Unfortunately, neither of the essential cases may be present in the forests $T_1',T_2'$, and therefore the ideas given above may not be applicable. However, they do work if we generalize our notions. First, we generalize the notion of ``active sibling pair in $T_1'$''.
\begin{definition}
A set of active leaves $U$ 
is an {\em active sibling set} in $T_1'$ if the leaves in $U$ are the only active leaves in the subtree of $T_1'$ rooted at $\lca_1( U )$.
$U$ is a {\em compatible active sibling set} in $T_1'$ if $U$ is an active sibling set in $T_1'$ that contains no inconsistent triplets.
\end{definition}
Note that we will only use the term compatible active sibling set for $T_1'$, and never for $T_2'$. We will therefore sometimes omit the reference to $T_1'$, and simply talk about a ``compatible active sibling set''. 

We similarly generalize the notion of a subtree in $T_1'$ containing exactly three active leaves that form an inconsistent triplet. 
\begin{definition}
A set of active leaves $R \cup B$ is a {\em minimal incompatible active sibling set} in $T_1'$ if $R\cup B$ is incompatible, $R$ and $B$ are compatible active sibling sets in $T_1'$, and $p_1( R ) = p_1( B )$.
\end{definition}
The Red-Blue Algorithm now proceeds as follows: it begins by identifying a minimal incompatible active sibling set $R\cup B$ in $T_1'$. Such a set can be found by checking if the active leaf sets of the left and right subtrees of the root are compatible sets. If yes, then either all active leaves are compatible, or we have found a minimal incompatible active set. If not, then the active leaf set of one of the subtrees is incompatible, and we recurse on this subtree until we find a node in $T_1'$ for which the active leaf sets of the left and right subtrees form a minimal incompatible set $R\cup B$. 
Note that we can assume $\lca_2(R)=\lca_2(R\cup B)$.

The algorithm will then ``distill'' $R$ by repeatedly considering sibling pairs $u,v$ in $R$, and executing operations similar to those in Algorithm~\ref{fig:3approx}, except that only one of $u$ and $v$ becomes inactive (and a bit more care has to be taken in certain cases). Procedure~\ref{fig:resolvepair} gives the procedure \resolvepair\ the algorithm uses for handling a sibling pair $u,v$.

\begin{procedure}[ht]
	\eIf{ $u$ and $v$ are in different trees in $T_2'$ }
	{	
		Relabel $u$ and $v$ if necessary so that $\lca_2(u,v)$ is not in the tree containing $u$ in $T_2'$.\; \label{alg:cond1}
		\eIf{ the tree containing $u$ in $T_2'$ has other active leaves not in $U$} 
		{
			{\it FinalCut:} Cut off $u$ in $T_1'$ and $T_2'$ and make it inactive. 
			\dual{$y_u \ot 1$.}\;  \label{alg:finalcut}
		}
		{
			Cut off  $u$ in $T_1'$ and make it inactive.  \dual{$y_u \ot 1$.} \;\label{alg:2cut1}
		}
	}
	{
		\eIf{ $u$ and $v$ are active siblings in $T_2'$ }
		{
			Merge $u$ and $v$ (i.e., make $u$ inactive to ``merge'' it with $v$). \; \label{alg:2merge}
		}
		{
			Relabel $u$ and $v$ if necessary so that $p_2(u)\neq \lca_2(u,v)$. \;
			Cut off an active subtree $W$ between $u$ and $v$ by cutting the edge below $p_2(u)$ that is not on the path from $u$ to $v$. \dual{Decrease $y_{p_2(u)}$ by $1$.} \;\label{alg:2cutW}
			\eIf{$u$ and $v$ are now active siblings in $T_2'$}
			{
				{\it Merge-After-Cut}: Merge $u$ and $v$ (i.e., make $u$ inactive to ``merge'' it with $v$). \dual{$y_u\ot 1$.}	\;	\label{alg:cutmerge} 
			}
			{				
				Cut off $u$ in $T_1'$ and $T_2'$ and make $u$ inactive. \dual{$y_{u} \ot 1$.}\;	\label{alg:cutu}
			}
		}
	}
\caption{1(): \resolvepair$(u,v)$}\label{fig:resolvepair}
\end{procedure}

Arguments similar to those in Section~\ref{sec:3approx} show that \resolvepair\ maintains dual feasibility, {\em provided that we initially reduce $y_{\lca_1(R)}$ by $1$}. It is also not hard to verify that \resolvepair\ increases the dual objective by at least half the increase in the primal objective, and the only thing that is therefore needed to show that the algorithm is a 2-approximation is that we can ``make up for'' the initial decrease of the dual objective caused by decreasing $y_{\lca_1(R)}$.
Let us define the operation of ``distilling'' $R$ as starting by reducing $y_{\lca_1(R)}$ by $1$, and then repeatedly finding a pair of active leaves $u,v$ in $R$ which are siblings in $T_1'$ and executing \resolvepair$(u,v)$ until only two active leaves $\hat u,\hat v$ in $R$ remain. Since all other leaves in $R$ are inactive, $\hat u$ and $\hat v$ form an active sibling pair in $T_1'$.

If pair $\hat u,\hat v$ is a ``success'' 
or if line~\ref{alg:finalcut} or~\ref{alg:cutmerge} was executed at least once during the distilling of $R$, then there exists an operation that makes at least one of $\hat u,\hat v$ inactive and updates the \iftoggle{abs}{dual}{dual solution}, so that the total increase in the primal objective is at most twice the total increase in the dual objective caused by the processing of pairs in $R$. 
Procedure~\ref{fig:2approxp1} gives the complete description of the procedure that, if successful, ``resolves'' set $R$ (and will return ``{\sc Success}''):
\begin{lemma}\label{lemma:successshort}
If \resolveset$(R)$ returns {\sc Success} then at least one leaf in $R$ became inactive, and the increase in the primal objective $|E(T_2)\setminus E(T_2')|$ caused by the procedure is at most twice the increase in the dual objective. 
\end{lemma}
\iftoggle{abs}{Lemma~12 of the full version~\cite{full}}{Lemma~\ref{lemma:success} in Section~\ref{sec:towards}} contains a more precise formulation of this lemma.

\begin{procedure}[ht]
\caption{2(): \resolveset$(R)$}\label{fig:2approxp1}

\dual{Decrease $y_{\lca_1( R )}$ by $1$.} \; \label{alg:initdual}
\While{there exist at least three active leaves in $R$}
{
	Find $u,v$ in $R$ that form an active sibling pair  in $T_1'$.\;\label{alg:act}
	\resolvepair$(u,v)$.\;
}
Let $\hat u,\hat v$ be the remaining active leaves in $R$.\;\label{alg:lastact}
\uIf{$\hat u$ and $\hat v$ are active siblings in $T_2'$\label{alg:iflastsib} }
{	
	Merge $\hat u$ and $\hat v$ (i.e., make $\hat u$ inactive to ``merge'' it with $\hat v$). \dual{$y_{\hat u}\ot 1$.} \; \label{alg:finalmerge}
	Return {\sc Success}.\;
}
\uElseIf{($\hat u$ and $\hat v$ are in different trees in $T_2'$) or (the tree containing $\hat u$ and $\hat v$ does not contain an active leaf $w$ such that $\hat u\hat v|w$ in $ T_2$) \label{alg:nofail}}
{		
	Cut off $\hat u$ in $T_2'$ (if $\hat u$'s tree contains at least one other active leaf) and in $T_1'$ and make $\hat u$ inactive.	\dual{$y_{\hat u}\ot 1$.}\; \label{alg:lastif}
	Cut off $\hat v$ in $T_2'$ (if $\hat v$'s tree contains at least one other active leaf) and in $T_1'$ and make $\hat v$ inactive.	\dual{$y_{\hat v}\ot 1$.}\; \label{alg:lastif2}
	Return {\sc Success}.\;
}
\uElseIf{(At least one {\em FinalCut} or {\em Merge-After-Cut} was executed in some call to \resolvepair\ in the course of the current \resolveset\ procedure) }
{
	\resolvepair$(\hat u,\hat v)$.\;\label{alg:lastif3a}
	Cut off the last active leaf $\hat v$ in $ U$ in $T_2'$ and in $T_1'$ and make $\hat v$ inactive.	\dual{$y_{\hat v}\ot 1$.}\; \label{alg:lastif3}
	Return {\sc Success}.\;	
}
\Else
{
	Return	{\sc Fail}.\;
}	
\label{alg:end}
\end{procedure}
If Lemma~\ref{lemma:successshort} applies, we have made progress (since we have made at least one leaf inactive), and we will have paid for the increase in the primal objective $|E(T_2)\setminus E(T_2')|$ caused by the procedure by twice the increase in the dual objective. 

Otherwise, the last active pair of leaves $\hat u,\hat v$ in $R$ remain active, and we will have a ``deficit'' in the sense that the increase in the dual objective is at most half the increase in the primal objective {\it plus} 1. In this case, we similarly distill $B$ by repeatedly calling \resolvepair$(u,v)$ for pairs $u,v$ in $B$ that are active siblings in $T_1'$ until only a single active leaf in $B$ remains. However, we will show that in order to retain dual feasibility, we do not need to start the distilling of $B$ by decreasing $y_{\lca_1(B)}$ (which would give a total ``deficit'' of 2), but that we can ``move'' the initial decrease of $y_{\lca_1(R)}$ to instead decrease $y_{\lca_1(R\cup B)}$. \iftoggle{abs}{Lemma~13 in the full version~\cite{full}}{Lemma~\ref{lemma:safe2} in Section~\ref{sec:patriot}} shows that this indeed preserves dual feasibility.

Once $R$ and $B$ have both been ``distilled'', we are left with $\hat u,\hat v,\hat w$ that are an inconsistent triplet and form the active leaf set of a subtree in $T_1'$. \iftoggle{abs}{We }{In Section~\ref{sec:patriot}, we then }show how to deal with the triplet $\{\hat u,\hat v,\hat w\}$ (in ways similar to those in Figure~\ref{fig:triplet}) and we prove that in the entire processing of $R\cup B$, we have increased the dual objective by half of the number of edges we cut from $T_2'$.

Algorithm~\ref{fig:2approx} gives an overview of the ``Red-Blue Algorithm''. It first calls a procedure \preprocess, which executes simple operations that do not affect the primal or dual objective: merging two leaves if they are active siblings in both forests, and cutting off and deactivating a leaf in $T_1'$ if it is the only active leaf in its tree in $T_2'$. At the end of an iteration, the Red-Blue algorithm needs to consider different cases for the final triplet. The description of these subroutines can be found \iftoggle{abs}{in \cite{full}}{in Section~\ref{sec:onetree} and Section~\ref{sec:multitrees}}.

\begin{algorithm}
\caption{Red-Blue Algorithm for Maximum Agreement Forest}\label{fig:2approx}

Set $T_1' \ot T_1$, $T_2' \ot T_2$, ${\cal L}' \ot {\cal L}$.  \dual{$y_u \ot 0$ for all $u \in {\cal L}$.} \;

\preprocess. \;

\While{${\cal L}' \neq \emptyset$}
{
	Find a minimal incompatible active sibling set $R \cup B$, with $\lca_2(R)=\lca_2(R\cup B)$.\;\label{alg:beginwhile}

	\If{ \resolveset$(R)$ returns {\sc Fail} }
	{
		\dual{Decrease $y_{\lca_1( R\cup B )}$ by $1$, and increase $y_{\lca_1(R)}$ by $1$.} \;\label{alg:movedual}
	
		\While{ there exist at least two active leaves in $B$ \label{alg:beginresolveB}}
		{
			Find $u,v$ in $B$ that form an active sibling pair in $T_1'$.\;
			\resolvepair$(u,v)$.\;
		}\label{alg:endresolveB}
				
		Let $\hat r_1,\hat r_2\in R$ and $\hat b\in B$ be the remaining active leaves. \;\label{alg:triplet}
		Consider three different cases depending on whether $\hat r_1,\hat r_2$ and $\hat b$ are in one, two or three different trees in $T_2'$ (see \iftoggle{abs}{Section~6.1 and~6.2 in~\cite{full}}{Section~\ref{sec:onetree} and~\ref{sec:multitrees}} for details).\;
	}
	\preprocess.\label{alg:endwhile}\;
}
\end{algorithm}

\begin{theorem}\label{thm:2approx}
The Red-Blue Algorithm is a 2-approximation algorithm for \iftoggle{abs}{the Maximum Agreement Forest problem}{MAF}.
\end{theorem}

\iftoggle{abs}{}{
\section{Distilling the Essence of a Compatible Active Sibling Set}\label{sec:towards}
In this section, we will prove (a more precise version of) Lemma~\ref{lemma:successshort}, i.e., that if \resolveset\ returns {\sc Success}, then we have made progress towards a feasible primal solution, and we have increased the dual objective by at least half the increase in the primal objective. 
Because our arguments for \resolvepair$(u,v)$ will not only be used for a pair $u,v\in R$, but also (in Section~\ref{sec:patriot}) for a pair $u,v\in B$, we will let $U$ denote an arbitrary compatible active sibling set $U$.

We begin by noting that our overall description of the algorithm is iterative, and that we thus assume that we have some global variables representing the forests $T'_1$ and $T'_2$, the set of active leaves ${\cal L}'$ and a setting of the dual variables $y$ that are modified by the procedures.
%To be precise, in addition to the inputs specified in Procedure~\ref{fig:resolvepair} and~\ref{fig:2approxp1}, the input implicitly also {\it valid} tuple $(T_1',T_2',{\cal L}',y)$:
\begin{definition}
We say a tuple $(T_1',T_2',{\cal L}',y)$ is {\em valid} for an instance $(T_1,T_2,{\cal L})$ of the Maximum Agreement Forest Problem, if
\begin{itemize}
\item
$T_1'$ and $T_2'$ can be obtained from $T_1$ and $T_2$ respectively by edge deletions,
\item
all leaves in ${\cal L}'$ are part of one tree in $T_1'$,
\item
$y_u = 0$ for all $u \in {\cal L}'$,
\item
the dual solution associated with $T'_2$, ${\cal L}'$, and $y$ is feasible for (D), 
\item
the inactive trees in $T_1'$ and $T_2'$ can be paired up into pairs of trees with the same compatible leaf set,
\item
for each active leaf $u\in {\cal L}'$, the subtree containing $u$ rooted at the child of $p_i(u)$ for $i=1,2$ contains the same leaf set in both forests, and this leaf set is compatible.
\end{itemize}
\end{definition}
Note that $(T_1,T_2,{\cal L},y_0)$ is valid for $(T_1,T_2,{\cal L})$ where $y_0( u ) = 0$ for all leaves $u$.

In order to prove that a tuple remains valid after calling \resolvepair$(u,v)$, and in particular, that the associated dual solution remains feasible, we need one additional notion.
\begin{definition}
For a compatible active sibling set $U\subseteq {\cal L}$, we call a tuple $(T_1',T_2',{\cal L}',y)$ {\em $U$-safe}, if for any compatible set of leaves $L$ and any tree with active leaf set $A$ in $T_2'$ the following holds: if $(L\cap A)\cap U\neq \emptyset$ and $(L\cap A)\setminus U\neq \emptyset$,
 then the load on $L$ is at most 0 in the dual solution associated with $T_2', {\cal L}'$ and $y$.
\end{definition}

The idea behind this definition is that the execution of \resolvepair\ on an active sibling pair $u,v$ in $T_1'$ with $u,v\in U$ will only increase the load on sets $L$ such that $L$ contains both $u$ and some leaf $w\not\in U$ in the tree in $T_2'$ that contained $u$. Hence, $U$-safeness implies that the load remains at most 1 for compatible sets, and hence the dual solution remains feasible. We will furthermore show that $U$-safeness is preserved when $u$ becomes inactive. 
 
 \begin{observation}\label{obs:makesafe}
If $U$ is an active sibling set in $T_1'$, then we can make a given valid tuple $U$-safe, by decreasing $y_{\lca_1(U)}$ or $y_{p_1(U)}$ by $1$. 
\end{observation}
In \resolveset$(U)$, we decrease $y_{\lca_1(U)}$ as this will be helpful if the final active sibling pair $\hat u,\hat v$ turns out to give a ``success'', but the flexibility implied by Observation~\ref{obs:makesafe} will prove useful later if \resolveset\  fails; see Section~\ref{sec:patriot}.

The following lemma shows that $U$-safeness will ensure that \resolvepair\ returns a valid tuple, and that this tuple is again $U$-safe.

\begin{lemma}\label{lemma:safe}
Let $(T_1',T_2',{\cal L}',y)$ be a valid tuple, let $U$ be a compatible active sibling set in $T_1'$, and let $u,v\in U$ be an active sibling pair in $T_1'$.
If $(T_1',T_2',{\cal L}',y)$ is $U$-safe, then the tuple $(\tilde T_1',\tilde T_2',\tilde {\cal L}',\tilde y)$  that results from the procedure \resolvepair$(u,v)$ is a valid tuple that is $U$-safe.
\end{lemma}
\begin{proof}
The first three properties of a valid tuple are clear from the description of the procedure. The last two properties follow from the fact that only active sibling pairs in both $T'_1$ and $T'_2$ are merged, and trees become inactive when they contain a single active leaf, which is exactly the same as in Section~\ref{sec:3correct}.  It remains to show that the modified dual solution is feasible and that it is $U$-safe.\footnote{In fact, we will prove a stronger property than $U$-safeness, namely, that the load on any compatible set $L$ is at most 0 if $L\cap U$ contains active leaves. In the next section, we will see why we need the weaker definition of $U$-safeness.}

We consider the line numbers in an execution of \resolvepair\ that may affect the load on a set $L$:
\begin{itemize}
\item[\ref{alg:finalcut}.] ({\it FinalCut}) Let $A$ be the active leaves in the tree in $T_2'$ containing $u$ at the start of the procedure. Then $z_{A}$ is decreased to 0, and $y_u$ and $z_{A\setminus\{u\}}$ are increased to 1. This increases the load of a set $L$ only if $L$ contains $u$ and an active leaf $w$ in $A\setminus\{u\}$. Note that $w\not \in U$ because $\lca_2(u,w)$ must be in the tree containing $u$ and $w$, and thus it would be a strict descendent of $\lca_2(u,v)$, contradicting the fact that $u,v$ are an active sibling pair in $T_1'$ and $U$ is compatible. Therefore, by the fact that the input tuple is $U$-safe, the dual solution remains feasible.

It remains to prove that the new dual solution is $U$-safe. The load only increased for sets $L$ that contain $u$ and an active leaf $w$ in $A\setminus\{u\}$. We will show that if $L$ is compatible, then $L$ cannot contain any other active leaf $u'\in U$, and hence, $L$ does not need to have load at most 0 for the tuple to be $U$-safe. Suppose by contradiction that $L$ contains $u$, $w\in A\setminus\{u\}$, and an active leaf $u'\in U$ with $u'\neq u$.  Since $U$ is a compatible active sibling set in $T_1'$, $u'u|w$ in $T_1$. On the other hand, $uw|v$ in $T_2$, because $\lca_2(u,v)$ is not on the path from $u$ to $w$ (by the condition in line~\ref{alg:cond1}), and thus also $uw|u'$ in $T_2$ because $u,v$ is an active sibling pair in $T_1'$ and $U$ is compatible. Hence $\{u,u',w\}$ is an inconsistent triplet and $L$ is thus not compatible.
\item[\ref{alg:2cut1}.] Since $u$ is the only active leaf in its tree in $T_2'$, $z_{\{u\}}$ decreases from $1$ to $0$. Therefore the load on any compatible set $L$ does not increase when $y_u$ is set to $1$. 
\item[\ref{alg:2cutW}.] Let $A$ be the set of active leaves in the tree containing $W$ in $T_2'$ before cutting off $W$. 
$z_A$ decreases by $1$, $z_{ A \setminus W }$  increases by $1$, $z_W$ increases by $1$.   The only sets $L$ for which the load potentially increases by $1$ are sets $L$ so that $W \cap L \neq \emptyset$ and $(A \setminus W) \cap L \neq \emptyset$. However, $p_2( W )=p_2(u) \in V[L]$ for such sets $L$, and since $y_{p_2(u)}$ is decreased by $1$, the load is not increased for any compatible set $L$, so the dual solution remains feasible and $U$-safe.
\item[\ref{alg:cutmerge}.] ({\it Merge-After-Cut}) Let $A$ be the set of active leaves in the tree containing $u$ and $v$. Since $u$ becomes inactive, $z_A$ decreases by $1$, $z_{A\setminus\{u\}}$ increases by $1$. $y_u$ is set to 1. Therefore, if the load on a set $L$ increases, then $u\in L$ and $L\cap (A\setminus\{u\})\neq \emptyset$. 

So let $w$ be an active leaf in $L\cap (A\setminus\{u\})$. 
Note that $V[L]$ must contain the internal node that was $p_2(u)$ in line~\ref{alg:2cutW}. Therefore, executing line~\ref{alg:2cutW} followed by line~\ref{alg:cutmerge} only increases the load for $L$ if it also contains an active leaf $w'\in W$, where $W$ is the set that was cut off in line~\ref{alg:2cutW}. 
Note that $uw'|w$ in $T_2$, and that $u\in U$ and $w'\not\in U$. Since $L$ must be compatible, this implies that $L$ cannot contain any active leaves in $U$ after executing line~\ref{alg:cutmerge}, and, in particular, that $w\not\in U$. Therefore, since $u,w$ are both in $A$ and the tuple was $U$-safe, the load on $L$ becomes at most 1, and thus the dual solution remains feasible. Moreover, if $L$'s load increased then $L$ contains no more active leaves in $U$, so the tuple remains $U$-safe.

\item[\ref{alg:cutu}.] Exactly the same arguments as for line~\ref{alg:cutmerge} apply. \qedhere
\end{itemize}

\end{proof}

\begin{observation} \label{obs:increase}
Suppose the load on a set $L$ increases during a call to \resolvepair$(u,v)$. Let $u$ be the leaf among $u$ and $v$ that is deactivated in this call. Then it holds that $u \in L$, there exists $w \in L \setminus U$ so that $uw|v$ in $T_2$, and $L \cap U$ does not contain any active leaves after the call to \resolvepair.
\end{observation}
\begin{proof}
This follows directly from the proof of the previous lemma: the only lines where the load on any set $L$ potentially increases are lines \ref{alg:finalcut} and \ref{alg:2cutW} followed by line~\ref{alg:cutmerge} or \ref{alg:cutu}. The observation follows directly from what is stated in the proof.
\end{proof}

We now consider the change in the objective value of the dual solution, i.e., $D(\tilde T_2',\tilde{\cal L}', \tilde y)-D(T_2',{\cal L}',y)$ and the objective value of the primal solution, i.e., $|E(T_2')\setminus E(\tilde T_2')|$. In Table~\ref{tab:resolvepair} we use $\Delta D$ to denote the change in $D(T_2',{\cal L}',y)$ caused by \resolvepair, and we denote by $\Delta P$ the change in $|E(T_2)\setminus E(T_2')|$.

\begin{table}[h!]
\begin{center} 
\begin{tabular}{|ll||c|c|}
\hline
&line number & $\Delta P$ & $\Delta D$\\
\hline
*&\ref{alg:finalcut} ({\it FinalCut}) &1&$1$\\
&\ref{alg:2cut1}&0&$0$\\
&\ref{alg:2merge}&0&$0$\\
*&\ref{alg:2cutW} followed by \ref{alg:cutmerge} ({\it Merge-After-Cut}) &1&$1$\\
&\ref{alg:2cutW} followed by \ref{alg:cutu}&2&$1$\\
\hline
\end{tabular}
\end{center}
\caption{$\Delta P$ denotes the change in $|E(T_2)\setminus E(T_2')|$ and $\Delta D$ denotes the change in $D(T_2',{\cal L}',y)$ for each possibility in Procedure~\ref{fig:resolvepair}: \resolvepair. }\label{tab:resolvepair}
\end{table}

We see that each possibility increases the number of edges cut from $T_2'$ by at most twice the increase in the dual objective value. 
Furthermore, the two possibilities marked with a star ({\it FinalCut} and {\it Merge-After-Cut}) have $\Delta P\le 2\Delta D-1$.

Now, let $U=R$ be a compatible active sibling set, and consider \resolveset$(R)$ as given in Procedure~\ref{fig:2approxp1}. It starts by decreasing $\lca_1(U)$ by 1 to make the tuple $U$-safe, and it then repeatedly calls \resolvepair, until only two active leaves in $U$ remain.   It continues to process these last two leaves only if it can guarantee that the total increase in the dual objective is at least $\tfrac 12$ times the increase in the number of edges cut from $T_2'$; in other words, if it can ``make up'' for the dual deficit that was created to make the initial tuple $U$-safe. In this case, the procedure outputs {\sc Success}, and otherwise it outputs {\sc Fail}. In the latter case, the procedure terminates with two leaves in $U$ still active.

We begin by showing that the tuple resulting from Procedure~\ref{fig:2approxp1} is valid.
\begin{lemma}\label{lemma:dualfeas}
Procedure \resolveset$(U)$ executed on a compatible active sibling set $U$ in $T_1'$ and a
valid tuple $(T_1',T_2',{\cal L}',y)$ outputs a valid tuple $(\tilde T_1',\tilde T_2',\tilde {\cal L}',\tilde y)$.  
\end{lemma}
\begin{proof}
The first three properties of a valid tuple are again clear from the description of the procedure. The last two properties follow from the arguments in Section~\ref{sec:3correct}.  It remains to show that the modified dual solution is feasible.

After executing line~\ref{alg:initdual}, the load on any set $L$ that contains $\lca_1(U)$ is decreased by $1$, so the tuple is $U$-safe. By Lemma~\ref{lemma:safe}, the tuple is still valid and $U$-safe after completion of the while-loop.

At this moment in the procedure, $\hat u$ and $\hat v$ are the only active leaves remaining in $U$. 
It is easily verified that the remainder of the procedure increases the load only for sets containing $\hat u$ or $\hat v$ and at least one other active leaf. 
By $U$-safeness, we have that a compatible set $L$ that contains $\hat u$ or $\hat v$ and least one active leaf $w\not \in U$ will have load at most 0. For a compatible set $L$ for which the active leaves are exactly $\hat u$ and $\hat v$, the load will be at most 0 as well: 
Note that $\lca_1(\hat u,\hat v)=\lca_1(U)$ (because at every execution of the while-loop an active sibling pair from $U$ is selected) and therefore, the load on a set $L$ containing $\hat u$ and $\hat v$ is at most 0 after line~\ref{alg:initdual}. Because $\hat u$ and $\hat v$ are still active, it follows from Observation~\ref{obs:increase} that the load on any set containing $\hat u$ or $\hat v$ has not increased by calling \resolvepair. 

Thus, to verify that the solution associated with $T_2', {\cal L}'$ and $y$ will be a dual feasible solution at the end of the procedure, it remains to verify that lines~\ref{alg:lastact}--\ref{alg:end} increase the load by at most 1 for any compatible set $L$.

If $\hat u,\hat v$ are active siblings in $T_2'$, then line~\ref{alg:finalmerge} is the last line executed by the algorithm that changes ${\cal L}'$ and $y$ and clearly, the load on any set increases by at most 1. If lines~\ref{alg:lastif3a}--\ref{alg:lastif3} are executed, then by Lemma~\ref{lemma:safe}, the dual remains feasible and the load is at most 0 on sets $L$ containing $\hat v$ after executing line~\ref{alg:lastif3a}, and hence the dual will remain feasible after executing line~\ref{alg:lastif3} as well.
Finally, suppose lines~\ref{alg:lastif} and~\ref{alg:lastif2} are executed, and
assume, by means of contradiction, that there is a compatible set $L$ for which the load increases by more than $1$.
Let $A_{\hat u}$ be the active leaves in the tree containing $\hat u$ before line~\ref{alg:lastif} and let $A_{\hat v}$ be the active leaves in the tree containing $\hat v$ (where it may be the case that $A_{\hat u}=A_{\hat v}$). $L$ must contain $\hat v$ and at least one active leaf $w\in A_{\hat v}\setminus \{\hat v\}$ (so that the load increases by $1$ in line~\ref{alg:lastif}), and $\hat u$ and at least one active leaf $w'\in A_{\hat u}\setminus \{\hat u,\hat v\}$ (so that the load increases by $1$ in line~\ref{alg:lastif2}). 
Now, if $A_{\hat u}=A_{\hat v}$, then the condition for line~\ref{alg:lastif} implies that it cannot be the case that $\hat u\hat v|w' $ in $T_2$, but this means $\{\hat u,\hat v,w'\}$ is an inconsistent triplet, contradicting the fact that $L$ is compatible. If $A_{\hat u}\neq A_{\hat v}$, then at most one of $A_{\hat u}, A_{\hat v}$ can contain leaves $q$ such that $\hat u\hat v|q$ in $T_2$, and hence, either $\{\hat u,\hat v,w\}$ or $\{\hat u,\hat v,w'\}$ is an inconsistent triplet, and again, the fact that $L$ is compatible is contradicted. 
\end{proof}

In order to show that, if the procedure outputs {\sc Success}, the total increase in the dual objective is at least $\tfrac12$ times the number of edges cut from $T_2'$, we need to eliminate some ``trivial'' cases first, see Procedure~\ref{fig:preprocess}.
\begin{procedure}
\caption{3(): \preprocess}\label{fig:preprocess}
\While{ (there exist active leaves $u,v$ that form an active sibling pair in both $T_1'$ and $T_2'$) {\em or} (there exists an active leaf $u$ that is the only active leaf in its tree in $T_2'$)}
	{
	\eIf{ there exist active leaves $u,v$ that form an active sibling pair in both $T_1'$ and $T_2'$}
	{		
		Merge $u$ and $v$ (i.e., make $u$ inactive to ``merge'' it with $v$).\;\label{alg:premerge}
	}
	{
		Let $u$ be the only active leaf in its tree in $T_2'$.\;
		Cut off $u$ in $T_1'$ (unless $u$ is the last active leaf), and make $u$ inactive. \dual{$y_u\ot 1$.}\;\label{alg:precut}
	}
}	
\end{procedure}

Note that the processing in Procedure~\ref{fig:preprocess}
cuts no edges from $T_2'$, and that the merge operation in line~\ref{alg:premerge} 
only merges active sibling pairs in both $T'_1$ and $T'_2$, which is exactly the same as in Section~\ref{sec:3correct}.  
Dual feasibility and the dual objective value are not affected by line~\ref{alg:precut}, since $z_{\{u\}}$ is decreased by $1$ when $y_u$ is increased by $1$.

\begin{lemma}\label{lemma:success}
Let $(T_1',T_2',{\cal L}',y)$ be a valid tuple, that has been preprocessed by procedure \preprocess, and let $U$ be an active sibling set in $T_1'$.
If \resolveset$(U)$ returns {\sc Success}, then it holds that 
\[D(\tilde T_2',\tilde{\cal L}',\tilde y)-D(T_2',{\cal L}',y) \ge \tfrac 12 \big(|E(T_2')\setminus E(\tilde T_2')|\big),\]
for the tuple $(\tilde T_1',\tilde T_2',\tilde {\cal L}',\tilde y)$ that is output by \resolveset$(U)$.  
\end{lemma}
\begin{proof}
From Table~\ref{tab:resolvepair}, we see that the execution of each line increases the number of edges cut from $T_2'$ $(\Delta P)$ by at most twice the increase in the dual objective value $(2\Delta D)$. If the line has a star, then $\Delta P=2\Delta D -1$. We also have that the first line of \resolveset, line~\ref{alg:initdual}, has $\Delta P=2\Delta D+2$, and we thus need to show that over the remainder of the procedure we ``make up for'' this initial decrease in the dual objective value by either executing two lines that have $\Delta P=2\Delta D-1$ or by executing a line that has $\Delta P=2\Delta D-2$.

If the algorithm returns {\sc Success}, then it has either executed (a) line~\ref{alg:finalmerge}, or (b) lines~\ref{alg:lastif} and~\ref{alg:lastif2}, or (c) lines~\ref{alg:lastif3a} and~\ref{alg:lastif3}. In case (a), we are done, since line~\ref{alg:finalmerge} had $\Delta P=0$ and $\Delta D=1$, so $\Delta P=2\Delta D-2$.
In case (c), we are also done: note that when executing line~\ref{alg:lastif3}, the last active leaf $\hat v$ is in a tree in $T_2'$ with some leaf $w$ such that $\hat u\hat v|w$ in $T_2$ (since otherwise, we would execute lines~\ref{alg:lastif} and~\ref{alg:lastif2}). Hence, cutting off $\hat v$ in $T_2'$ gives $\Delta P=1$ and $\Delta D=1$ (and making $\hat v$ inactive and setting $y_{\hat v}\ot 1$ does not effect the dual objective). Since line~\ref{alg:lastif3} is only executed if at least one starred line from Table~\ref{tab:resolvepair} was executed, we thus executed at least two lines that have $\Delta P=2\Delta D-1$ in total. 

The only remaining case is the case when the algorithm terminates by executing lines~\ref{alg:lastif} and~\ref{alg:lastif2} on pair $\hat u,\hat v$. 
Note that if $\hat u$ and $\hat v$ are in the same tree in $T_2'$ when executing lines~\ref{alg:lastif} and~\ref{alg:lastif2}, then this tree contains at least one other active leaf (since $\hat u$ and $\hat v$ are not siblings), and hence the last two lines together have $\Delta P=2$ and $\Delta D=2$, so $\Delta P=2\Delta D-2$.
The subtle issue if $\hat u$ and $\hat v$ are not in the same tree in $T_2'$ is that if $\hat u$ (or $\hat v$) is the only active leaf in its tree in $T_2'$, then line~\ref{alg:lastif} (respectively line~\ref{alg:lastif2}) has $\Delta P=\Delta D=0$, and hence this does not help to ``make up for'' the initial decrease in the dual objective value. If $\hat u$ (or $\hat v$) is not the only active leaf in its tree in $T_2'$, then line~\ref{alg:lastif} (respectively line~\ref{alg:lastif2}) has $\Delta P=2\Delta D-1$.

To analyze the case when $\hat u$ and $\hat v$ are in different trees in $T_2'$, let $U_L$ and $U_R$ be the active leaf sets of the subtrees of $T_1'$ rooted at the two children of $\lca_1(U)$ at the start of the procedure. 
Note that while there are at least three active leaves in $U=U_L\cup U_R$, an active sibling pair in $T_1'$ consisting of two leaves in $U$ will contain either two leaves in $U_R$ or two leaves in $U_L$. Therefore, each call to \resolvepair\ has as its arguments two leaves that are either both in $U_L$ or both in $U_R$, and the last two leaves satisfy $\hat u\in U_L, \hat v\in U_R$.

We have the following claim.
\begin{claim}\label{claim:bicolor}
Let $(T'_1,T_2',{\cal L}',y)$ be a valid tuple that has been preprocessed using Procedure~\ref{fig:preprocess}, and let $\tilde U$ be a compatible active sibling set in $T_1'$. If repeated calls to \resolvepair$(u,v)$ with $u,v\in \tilde U$, where $u,v$ are active siblings in $T_1'$ at the moment of the call, result in having a single active leaf $\tilde u$ in $\tilde U$, which is the only active leaf in its tree in $T_2'$, then a {\it FinalCut} or {\it Merge-After-Cut} must have been performed in one of the calls to \resolvepair.
\end{claim}
Note that the claim can be applied using $\tilde U=U_L$ and $\tilde U=U_R$ if $\hat u$, respectively $\hat v$, is the only active leaf in its tree in $T_2'$.
Hence, if $\hat u$ and $\hat v$ are in different trees in $T_2'$, then both $U_L$ and $U_R$ contribute at least one operation that has $\Delta P=2\Delta D-1$ as required, and thus the total increase in the number of edges cut from $T_2'$ is at most twice the total increase in the dual objective value.

We conclude by proving the claim.

\begin{proof_of_claim}
We will call an active tree in $T_2'$ $\tilde U$-unicolored, if all its active leaves are in $\tilde U$, and $\tilde U$-bicolored if it contains active leaves in $\tilde U$ and active leaves not in $\tilde U$.
Note that, initially, there must have been at least one active tree in $T_2'$ that was $\tilde U$-bicolored: otherwise, we could have preprocessed the tuple further in Procedure~\ref{fig:preprocess}. 
Let $A$ be the active leaf set of this tree. Now, either $\tilde u\in A$, or all leaves in $A\cap \tilde U$ will be inactive at the moment when $\tilde u$ is the only active leaf remaining in $\tilde U$. It then follows from the following observation that at least one {\it FinalCut} or {\it Merge-After-Cut} must have been performed if the remaining leaf $\tilde u$ is in a $\tilde U$-unicolored tree.
\begin{observation}\label{obs:bicolor}
Let $A$ be the active leaf set of some active tree in $T_2'$ that is $\tilde U$-bicolored.
If after a call to \resolvepair$(u,v)$ with $u,v\in \tilde U$, where $u,v$ are active siblings in $T_1'$, 
all leaves in $A\cap \tilde U$ are either inactive, or in a $\tilde U$-unicolored tree, then the 
\resolvepair\ procedure must have performed a {\it FinalCut} or {\it Merge-After-Cut}.
\end{observation}
To verify the observation, we consider the other possible executions of \resolvepair. Line~\ref{alg:2cut1} only deactives a leaf that is in a $\tilde U$-unicolored tree. Line~\ref{alg:2merge} deactivates a leaf in $\tilde U$ and does not affect whether the tree containing this leaf is $\tilde U$-bicolored or not.  Line~\ref{alg:2cutW} cuts off leaves that are not in $\tilde U$ from the tree containing the leaves $u$ and $v$. The remaining tree containing $u$ and $v$ is not $\tilde U$-unicolored, unless the procedure performs a {\it Merge-After-Cut}.
\end{proof_of_claim}
\end{proof}

\section{The Red-Blue Algorithm}\label{sec:patriot}
In the previous section we showed a procedure to resolve certain compatible active sibling sets. 
The Red-Blue Algorithm (see Algorithm~\ref{fig:2approx} in Section~\ref{sec:key}) uses this procedure as a subroutine: it starts by finding a ``minimal incompatible active sibling set'' $R\cup B$ and calls \resolveset$(R)$. 
In this section, we give more details on the Red-Blue Algorithm and a complete analysis of its correctness and approximation ratio.

First of all, note that we can always find a minimal incompatible active sibling set $R \cup B$ (if not all leaves are inactive after preprocessing using the \preprocess\  procedure): Consider $\lca_1({\cal L}')$. If the active leaf sets of the left and right subtree of this node are compatible, then let $R$ and $B$ be these two sets. Note that $R\cup B$ must be incompatible, since otherwise the \preprocess\  procedure would be able to make all leaves in $R\cup B$ inactive. If, on the other hand, the active leaf set of one of the subtrees is incompatible, we can recurse on this subtree until we find a node in $T_1$ for which the active leaf sets of the left and right subtrees are compatible. 

We assume without loss of generality that $\lca_2(R)=\lca_2(R\cup B)$. 
A property that we will use in our analysis and that explains the distinction between sets $R$ (the ``red leaves'') and $B$ (the ``blue leaves'') is the following:
\begin{observation}\label{obs:fail}
Let $R\cup B$ be a minimal incompatible active sibling set such that $\lca_2(R) = \lca_2(R\cup B)$.
Suppose \resolveset$(R)$ returns {\sc Fail}, and let $\hat r_1,\hat r_2$ be the remaining active leaves in $R$. Then 
\begin{enumerate}
\item
$\{\hat r_1,\hat r_2,v\}$ is an inconsistent triplet for any $v\in B$,
\item
$\hat r_1,\hat r_2$ are in the same tree in $T_2'$,
\item
there exists an active leaf $w\not\in R\cup B$ that is in the same tree in $T_2'$ as $\hat r_1,\hat r_2$, where $w$ is not a descendent of $\lca_2(R\cup B)$ (i.e., $w$ satisfies $uv|w$ in $T_2$ for all $u,v\in R\cup B$), 
\item
there exists an active leaf $x\not\in R$ that is in the same tree in $T_2'$ as $\hat r_1,\hat r_2$, where $x$ is a descendent of $\lca_2(R\cup B)$ (i.e., $x$ satisfies $\hat r_1 x|\hat r_2$ in $T_2$ or $\hat r_2x|\hat r_1$ in $T_2$).
\end{enumerate}
\end{observation}
Since $\lca_2(\hat r_1,\hat r_2)=\lca_2(R)=\lca_2(R\cup B)$, we have that for any $v\in B$ it is not the case that $\hat r_1\hat r_2|v$ in $T_2$, so $\{\hat r_1,\hat r_2,v\}$ is inconsistent. The fact that $\hat r_1, \hat r_2$ are in one tree in $T_2'$ follows from the fact that condition in line~\ref{alg:nofail} of \resolveset\ did not hold since otherwise \resolveset\ would have returned {\sc Success}. The existence of an active leaf $w$ with the stated properties follows from the same fact: Note that the path in $T_2'$ connecting $\hat r_1$ and $\hat r_2$ contains $\lca_2(R)=\lca_2(R\cup B)$, and that the leaf $w$ in line~\ref{alg:nofail} is not a descendent of this node. Hence $uv|w$ in $T_2$ for every $u,v\in R\cup B$.
The existence of $x$ with the stated properties follows from the fact that the condition in line~\ref{alg:iflastsib} did not hold.\\

If \resolveset$(R)$ returns {\sc Fail}, then the Red-Blue Algorithm increases $\lca_1(R)$ by 1 and decreases $\lca_1(R\cup B)$. Note that by Observation~\ref{obs:makesafe}, we could have initially made the tuple both $R$-safe {\it and} $B$-safe by decreasing $y_{p_1(R)}=y_{\lca_1(R\cup B)}$ instead of $y_{\lca_1(R)}$, so by Lemma~\ref{lemma:safe}, the tuple we have after \resolveset$(R)$ fails is valid and $\{\hat r_1,\hat r_2\}$-safe.
The algorithm then proceeds to repeatedly call \resolvepair$(u,v)$ for active sibling pairs $u,v\in B$ until a single active leaf $\hat b$ in $B$ remains. It follows from Observation~\ref{obs:fail} that the final three active leaves $\hat r_1,\hat r_2$ and $\hat b$ form an inconsistent triplet, which form an active subtree in $T_1'$. In Lemma~\ref{lemma:safe2} we show that the tuple we have at this moment is valid and $\{\hat r_1,\hat r_2\}$-safe and $\{\hat b\}$-safe. The next three subsections then explain how to deal with this triplet, depending on whether the leaves are all in one, two or three trees in $T_2'$.

\begin{lemma}\label{lemma:safe2}
Let $(T_1',T_2',{\cal L}',y)$ be a valid tuple at the start of line~\ref{alg:beginwhile}, and suppose \resolveset$(R)$ returns {\sc Fail}.
Then, the tuple in line~\ref{alg:triplet} is valid and $\{\hat r_1,\hat r_2\}$-safe and $\{\hat b\}$-safe.
\end{lemma}

\begin{proof}
Note that the fact that \resolveset$(R)$ fails means that the only operations that have been executed are repeated calls to \resolvepair$(u,v)$, with $u,v\in R$ or $u,v\in B$.

We first show that the tuple is valid, $\{\hat r_1,\hat r_2\}$-safe and $\{\hat b\}$-safe after executing line~\ref{alg:movedual}. We then show this continues to hold when executing lines~\ref{alg:beginresolveB}--\ref{alg:endresolveB}.

By the fact that $R$ is a compatible active sibling set in $T_1'$, it follows from Lemma~\ref{lemma:safe} that the tuple $(\tilde T_1',\tilde T_2',\tilde{\cal L}',\tilde y)$ resulting at the end of \resolveset$(R)$ is valid and $\{\hat r_1,\hat r_2\}$-safe. In fact, by Observation~\ref{obs:increase}, the only sets $L$ for which the load has increased in the course of \resolveset$(R)$ are sets containing a leaf in $R$ and a leaf not in $R$, and thus ``moving up the $-1$'' in line~\ref{alg:movedual} does not affect the dual feasibility and the $\{\hat r_1,\hat r_2\}$-safeness. 

We now show that it is also $\{\hat b\}$-safe. Note that the only operations that have been executed are calls to \resolvepair$(u,v)$ with $u,v\in R$, and that the call did not perform a {\it FinalCut}, as in that case \resolveset$(R)$ would return {\sc Success}. We show that these operations cannot increase the load on a compatible set $L$ that contains $\hat b$ and at least one other active leaf $x\not\in B$ that is in the same tree as $\hat b$ in $T_2'$. This proves that the tuple is $\{\hat b\}$-safe after line~\ref{alg:movedual}, since executing this line decreases the load on any set containing $\hat b$ and a leaf $x\not\in B$.

Since we know the procedure did not execute a {\it FinalCut}, and since the load on $L$ does not increase if $u,v$ are active siblings in $T_2'$, we only need to consider the case where the procedure executes line~\ref{alg:2cutW} followed by line~\ref{alg:cutmerge} or line~\ref{alg:cutu}. 
Suppose by contradiction that $L$ is compatible, the load on $L$ increases, and $L$ contains $\hat b\in B$ and an active leaf $x\not\in B$ that are in the same tree in $T_2'$ after this operation. By Observation~\ref{obs:increase}, if the load for $L$ increases, then $L$ contains $u$ and a leaf $w\not\in R$. Note that $u$ and $w$ are in an inconsistent triplet with any active leaf in $R$, because $uw|v$ in $T_2$ and $u$ and $v$ are active siblings in $T'_1$. Therefore, it must be the case that $x\not \in R$.

We discern three cases based on the relative position of $\lca_2( u, b )$ and $\lca_2( u, w )$ on the path from $u$ to the root of $T_2$. 

\noindent
Case~1: $\lca_2( u, b ) = \lca_2( u, w )$. Note that $\lca_2( x, b )$ must be a descendent of $\lca_2( u, w )$, because the edge below $\lca_2( u, w )$ towards $w$ was cut, and $x$ and $b$ are still in the same tree in $T_2'$.  Thus $bx|u$ in $T_2$, which contradicts that $L$ is compatible because $x \not\in R\cup B$.

\noindent
Case~2: $\lca_2( u, w )$ is a descendent of $\lca_2( u, b )$. Then $uw|b$ in $T_2'$, again contradicting that $L$ is compatible because $w \not\in R$.

\noindent
Case~3: $\lca_2( u, b )$ is a descendent of $\lca_2( u, w )$. Then $b$ is not in the same tree as $u$ in $T_2'$ at the start of the procedure because $\lca_2( u, w ) = p_2( u )$. Therefore, $bx|u$ in $T_2$, because $x$ is in the same tree as $b$ in $T_2'$, again contradicting that $L$ is compatible because $x \not\in R\cup B$.

It remains to consider the effect of executing \resolvepair$(u,v)$ for $u,v\in B$. Since $B$ is a compatible active sibling set in $T_1'$, by Lemma~\ref{lemma:safe} the tuple remains valid, and by Observation~\ref{obs:increase} it does not increase the load on any set containing $\hat b$. We now consider the effect on a set $L$ containing $r\in \{\hat r_1, \hat r_2\}$ and at least one other active leaf $x\not\in \{\hat r_1,\hat r_2\}$ that is in the same tree as $r$ in $T_2'$. If $x\in B$, then the load on $L$ did not increase by Observation~\ref{obs:increase}, so assume instead that $x\not \in B$. Morever, if the load on a set $L$ increases by executing \resolvepair$(u,v)$ for $u,v\in B$, then $L$ must contain $u$ and some leaf $w\not\in B$ that was in the same tree as $u$ at the start of \resolvepair$(u,v)$, and which satisfies $uw|v$ in $T_2$.

We show that $\{u,r,w,x\}$ contains an inconsistent triplet by considering three cases. If $r$ and $u$ were in different trees in $T_2'$ after the execution of \resolveset$(R)$, then for any $w$ that is in the same tree as $u$ in $T_2'$, it holds that $w\not \in \{\hat r_1,\hat r_2\}$ (by Observation~\ref{obs:fail} (2)), and that $uw|r$ in $T_2$, since $\lca_2(u,r)$ must be on the path from $r$ to the root, and hence $\lca_2(u,r)$ is in the tree containing $r$ and $\lca_2(B\cup R)$ in $T_2'$. Hence, in this case, $\{u,w,r\}$ form an inconsistent triplet, since $ur|w$ in $T_1$. If $r$ and $u$ were in the same tree in $T_2'$ after the execution of \resolveset$(R)$, and they are still in the same tree in $T_2'$ after \resolvepair$(u,v)$, then $w\neq r$ and $uw|r$ in $T_2$ (since $w$ must be a leaf in the set $W$ that is cut off by cutting the edge below $p_2(u)$), and thus again $\{u,w,r\}$ is an inconsistent triplet. Finally, if $r$ and $u$ were in the same tree in $T_2'$ after the execution of \resolveset$(R)$, but some subsequent call to \resolvepair$(u',v')$ separates them in $T_2'$, then $r$ and $x$ must be in the set $W$ that is cut off by \resolvepair$(u',v')$. But then we have that $rx|u''$ in $T_2$ for any leaf $u''\in B$ that is active at that time. Hence $rx|u$ in $T_2$, and thus $\{r,x,u\}$ is an inconsistent triplet, since $ru|x$ in $T_1$.

We have thus shown that the load cannot increase on a compatible set $L$ that contains $r\in \{\hat r_1,\hat r_2\}$ and at least one other active leaf $x\not\in \{\hat r_1,\hat r_2\}$ that is in the same tree as $r$ in $T_2'$. Therefore, the tuple remains $\{\hat r_1,\hat r_2\}$-safe throughout lines~\ref{alg:beginresolveB}--\ref{alg:endresolveB}.
\end{proof}

\subsection{Inconsistent triplet in a single tree in $T_2'$}\label{sec:onetree}

Suppose $\hat r_1,\hat r_2$ and $\hat b$ are in the same tree in $T_2'$ in line~\ref{alg:triplet}. 
Note that we are then exactly in case (I) of Section~\ref{sec:key}, and the triplet can either be in the configuration of subcase (I)(a) or of subcase (I)(b); see Figure~\ref{fig:triplet}. We give a formal description of the procedure for dealing with this case in Algorithm~\ref{fig:onetree}.

Note that the execution of lines~\ref{alg:beginresolveB}--\ref{alg:endresolveB} can never delete an edge from $T_2'$ that is above $\lca_2(R\cup B)$, because \resolvepair$(u,v)$ only cuts edges below $\lca_2(u,v)$, and hence in lines~\ref{alg:beginresolveB}--\ref{alg:endresolveB} only edges below $\lca_2(B)$ are cut.  Combined with Observation~\ref{obs:fail}, this implies that the tree in $T_2'$ containing $\hat r_1,\hat r_2$ and $\hat b$ contains at least one active leaf $w\not\in R\cup B$ that is not a descendent of $\lca_2(R\cup B)$.

\begin{procedure}
\caption{4a(): $\hat r_1,\hat r_2$ and $\hat b$ are in the same tree in $T_2'$}\label{fig:onetree}
Relabel $\hat r_1,\hat r_2$ if necessary so that $\hat b\hat r_1|\hat r_2$ in $T_2$.\;
Cut off the subtree rooted at $\lca_2(\hat r_1, \hat b)$ in $T_2'$. \dual{Decrease $y_{\lca_2(\hat r_1, \hat b)}$ by $1$.}\label{alg:cutlca}\;
Cut off $\hat r_2$ in $T_2'$ and $T_1'$ and make $\hat r_2$ inactive. \dual{$y_{\hat r_2}\ot 1$.}\label{alg:cutr2}\;
\eIf{ $\hat r_1,\hat b$ are now active siblings in $T_2'$ }
{
	Merge $\hat b$ and $\hat r_1$ (i.e., make $\hat b$ inactive to ``merge'' it with $\hat r_1$). \dual{$y_{\hat b}\ot 1$.} \label{alg:supermerge}\; 
}
{	
	Cut off $\hat r_1$ and $\hat b$ in $T_2'$ and $T_1'$ and make them inactive. \dual{$y_{\hat r_1}\ot 1$, $y_{\hat b}\ot 1$.}\label{alg:supercut}\;
}	
\end{procedure}

% % \newcommand{\arbitrarytreeat}[1]{\draw #1 -- ++(-1,-2) -- ++(2,0) -- ++(-1,2);}
% % \newcommand{\smallarbitrarytreeat}[1]{\draw #1 -- ++(-.5,-1) -- ++(1,0) -- ++(-.5,1);}
% % % offsets to use: -.3 -.6
% % \newcommand{\drawnode}[1]{\draw [fill] #1 circle [radius=.07];}
% % \newcommand{\treelabel}[3]{\node at (#1,{#2-.7}) {#3};} 
% \begin{figure}
% \begin{center}
% \input{incompatibleonetree.tex}
% \end{center}
% \vspace*{-1.5\baselineskip}
% \caption{Illustration of procedure~\ref{fig:onetree} in $T_2'$: Circled values denote the dual variables that are set by the algorithm, and bold diagonal lines denote edges that are cut. Triangles denote subtrees with active leaves (not all of these subtrees have to be nonempty).}
% \end{figure}

\begin{lemma}\label{lemma:onetree}
Let $(T_1',T_2',{\cal L}',y)$ be a valid tuple, that has been preprocessed by procedure \preprocess, and 
let $R\cup B$ be a minimal incompatible active sibling set with $\lca_2(R)=\lca_2(R\cup B)$, for which \resolveset$( R )$ returns {\sc Fail}.
If, after executing lines~\ref{alg:beginwhile}--\ref{alg:endresolveB}, the three remaining active leaves are in a single tree in $T_2'$, then the tuple $(\tilde T_1',\tilde T_2',\tilde {\cal L}',\tilde y)$ 
after executing lines~\ref{alg:beginwhile}--\ref{alg:endresolveB} followed by Procedure~\ref{fig:onetree} is valid, and satisfies
\[D(\tilde T_2',\tilde{\cal L}',\tilde y)-D(T_2',{\cal L}',y) \ge \tfrac 12 \big(|E(T_2')\setminus E(\tilde T_2')|\big).\]
\end{lemma}
\begin{proof}
The first three properties of a valid tuple are again clear from the description of the procedure. The last two properties follow from the arguments in Section~\ref{sec:3correct}.  It remains to show that the modified dual solution is feasible, and that the increase in the dual objective value can ``pay for'' half of the increase in the primal objective value.

Letting, as in line~\ref{alg:triplet}, $\hat r_1,\hat r_2$ be the remaining active leaves in $R$ and $\hat b$ the remaining active leaf in $B$, we have by Lemma~\ref{lemma:safe2} that the tuple is valid and $\{\hat r_1,\hat r_2\}$-safe and $\{\hat b\}$-safe after lines~\ref{alg:beginwhile}--\ref{alg:endresolveB}.

We consider what happens to the dual solution when executing Procedure~\ref{fig:onetree}. The numbers refer to the line numbers in the procedure.
\begin{itemize}
\item[\ref{alg:cutlca}.] 
Let $A_1,A_2$ be the active leaf sets of the two trees created, with $\hat b, \hat r_1\in A_1, \hat r_2\in A_2$. Then $z_{A_1\cup A_2}$ decreases by $1$ and $z_{A_1}$ and $z_{A_2}$ increase by $1$. Note that $y_{\lca_2(\hat r_1,\hat b)}$ is decreased by $1$.
Since any set $L$ with $L\cap A_1\neq\emptyset$ and $L\cap A_2\neq\emptyset$ has $\lca_2(\hat r_1,\hat b)\in V[L]$, this therefore does not increase the load on any set $L$. The objective value of the dual solution is also not changed.
\item[\ref{alg:cutr2}.] 
Let $A_2$ be the active leaves of the tree in $T_2'$ containing $\hat r_2$ before line~\ref{alg:cutr2} is executed. Note that $A_2 \setminus \{\hat r_2 \}$ is not empty, because it contains a node that is not a descendent of $\lca_2(R\cup B)$ (see Observation~\ref{obs:fail}). Therefore line~\ref{alg:cutr2} increases $z_{A_2\setminus \{\hat r_2\}}$ by $1$; it also decreases $z_{A_2}$ by $1$ and increases $y_{\hat r_2}$ by 1. 
This increases the load on sets $L$ containing $\hat r_2$ and at least one leaf $w\in A_2\setminus \{\hat r_2\}$. Note that $w\neq \hat r_1$, since $\hat r_1\not\in A_2$. By the fact that the dual solution is $\{\hat r_1,\hat r_2\}$-safe, we thus know that the load on any set $L$ for which the load increases had load at most 0 prior to the increase. 
The dual objective value is increased by $1$.
\item[\ref{alg:supermerge}.]
If $\hat b$ and $\hat r_1$ are active siblings in $T_2'$, they are the only active leaves in their tree in $T_2'$. Hence, line~\ref{alg:supermerge} decreases $z_{\{\hat r_1,\hat b\}}$ by $1$, and it increases $y_{\hat b}$ and $z_{\{\hat r_1\}}$ by $1$. The load is increased only on sets $L$ containing both $\hat b$ and $\hat r_1$.  Such sets $L$ had load at most 0 at the start of the procedure by Lemma~\ref{lemma:safe2}. Furthermore, $L$ cannot have had its load increased in line~\ref{alg:cutr2}, as this would mean $L$ contains inconsistent triplet $\{\hat b, \hat r_1,\hat r_2\}$. Hence, the dual solution remains feasible.
The dual objective value is increased by $1$.
\item[\ref{alg:supercut}.]
The value of $z_{A_1}$ is decreased by $1$, and $y_{\hat b}, y_{\hat r_1}$ and $z_{A_1\setminus \{\hat b,\hat r_1\}}$ are increased by $1$. 
This increases the load on sets $L$ containing leaves in at least two of the sets $\{\hat b\}, \{\hat r_1\}, A_1\setminus\{\hat b,\hat r_1\}$. 
Furthermore, note that a compatible set $L$ can contain leaves in at most two of these sets, and thus the load on a compatible set $L$ increases by at most 1.

If the load on a compatible set $L$ increases, then $L$ cannot contain $\hat r_2$: the load on $L$ increasing implies that $L$ contains $\hat b$ or $\hat r_1$, and at least one other leaf in $A_1$. These two nodes, say $u,v$, in $A_1$ and $\hat r_2$ are an inconsistent triplet: $uv| \hat r_2$ in $T_2$, but $u,v$ are not both in $R$ and they are not both in $B$, so it cannot be the case that $uv|\hat r_2$ in $T_1$. This shows that $L$ is incompatible. 

Therefore the load on $L$ was at most 0 at the start of the procedure (by Lemma~\ref{lemma:safe2}) and the load has not increased by line~\ref{alg:cutr2}. Hence, the dual solution remains feasible.
The dual objective value is increased by $2$.
\end{itemize}

We now consider the total change in the primal and dual objective value.
Let $\Delta P_1$ be the number of edges in $E(T_2')\setminus E(\tilde T_2')$  due to lines~\ref{alg:beginwhile}--\ref{alg:endresolveB}, and let $\Delta P_2$ be the number of edges in $E(T_2')\setminus E(\tilde T_2')$  due to Procedure~\ref{fig:onetree}. 
Similarly, let $\Delta D_1$ be the total change in the dual objective value by lines~\ref{alg:beginwhile}--\ref{alg:endresolveB}, and
$\Delta D_2$ the change in the dual objective due to Procedure~\ref{fig:onetree}.

We have $\Delta D_1 \ge \frac 12 \Delta P_1 -1$ by taking into account the initial decrease in the dual objective value and Table~\ref{tab:resolvepair}.
Note that $\Delta P_2= 2$ if line~\ref{alg:supermerge} is executed, and $\Delta P_2= 4$ if line~\ref{alg:supercut} is executed.
The arguments about the dual solution given above also show that $\Delta D_2= 2$ in the first case, and $\Delta D_2= 3$ in the second case.
Hence, we have that $\Delta D_2\ge \frac 12 \Delta P_2+1$, and thus $\Delta D_1+\Delta D_2\ge \frac 12\left(\Delta P_1+\Delta P_2\right)$.
\end{proof}
\subsection{Inconsistent triplet in multiple trees in $T_2'$}\label{sec:multitrees}
We give the procedures for dealing with the remaining cases in Procedure~\ref{fig:threetrees} and Procedure~\ref{fig:twotrees}. 
These are more complicated than case (II)(a) that was shown in Figure~\ref{fig:triplet} in Section~\ref{sec:key}. The reason for this additional complexity is that the calls to \resolvepair$(u,v)$ with $u,v\in R\cup B$, may lead to $\hat r_1,\hat r_2$ or $\hat b$ being the only active leaf in their respective trees in $T_2'$.
 If this happens, we may need to identify two inactive trees at the end of the procedure and ``retroactively merge'' them. 
In Figures~\ref{fig:retrofig} and~\ref{fig:retrofig2}, we give examples that illustrate the retroactive merge for Procedure~\ref{fig:twotrees}.

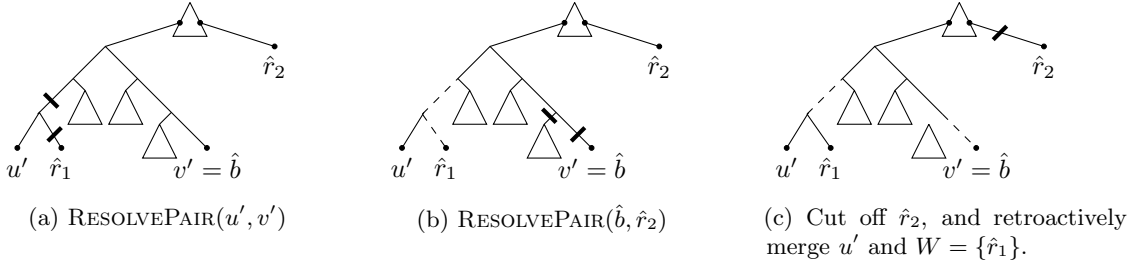
\begin{figure}
\begin{center}
\begin{subfigure}[t]{.26\textwidth}
\begin{tikzpicture}[scale=.45]
\draw (-2,-2) -- (0,0) -- (3,-3);
\draw (-2.6,-3) -- (-1.95,-1.95) -- (-1.3,-3);
\drawnode{(-2.6,-3)}
\node [below] at (-2.6,-3) {$u'$};
\drawnode{(3,-3)}
\node [below] at (3,-2.87) {$v' = \hat b$};
\drawnode{(-1.3,-3)}
\node [below] at (-1.3,-3) {$\hat r_1$};

%\node [below] at (-3,-3) {$\hat u'$};
%\drawnode{(3,-3)}
%\node [below] at (3,-2.83) {$\hat v'$};
%\drawnode{(0,0)}

\smallarbitrarytreeat{(-.6,-1.3)}
%\drawnode{(-.6,-1.3)}
%\smallarbitrarytreeat{(-1.6,-2.3)}
%\drawnode{(-1.6,-2.3)}
%\draw (-1.6,-2.3) -- (-1.95,-1.95);
\draw (-.6,-1.3) -- (-.95,-.95);
%\drawnode{(-1.95,-1.95)}
%\drawnode{(-.95,-.95)}

\smallarbitrarytreeat{(.6,-1.3)}
%\drawnode{(-.6,-1.3)}
\smallarbitrarytreeat{(1.6,-2.3)}
%\drawnode{(-1.6,-2.3)}
\draw (1.6,-2.3) -- (1.95,-1.95);
\draw (.6,-1.3) -- (.95,-.95);
%\drawnode{(-.95,-.95)}

\smallarbitrarytreeat{(2.5,1.3)}
\draw (0,0) -- (2.2,0.7);
\drawnode{(2.2,0.7)}
\drawnode{(2.8,.7)}
\draw (5,0) -- (2.8,.7);
\drawnode{(5,0)}
\node [below] at (5,0) {$\hat r_2$};

%\smallarbitrarytreeat{(5,2)}
%\draw (0,0) -- (4.7,1.4);
%\drawnode{(4.7,1.4)}
%\drawnode{(5.3,1.4)}
%\draw (10,-3) -- (5.3,1.4);
%\drawnode{(10,-3)}
%\node [below] at (10,-3) {$\hat r_2$};

% cuts and duals

% cut u
%\draw [ultra thick] (-2.2,-2.8) -- (-2.6,-2.4); 
\draw [ultra thick] (-1.7,-2.8) -- (-1.3,-2.4); 
\draw [ultra thick] (-1.4,-1.8) -- (-1.8,-1.4); 
% cut v
%\draw [ultra thick] (2.4,-2.8) -- (2.8,-2.4); 
% cut w
%\draw [ultra thick] (9.4,-2.8) -- (9.8,-2.4); 
% cut tree
%\draw [ultra thick] (2,-2.35) -- (1.6,-1.95); 
% cut tree
%\draw [ultra thick] (1.5,.2) -- (1.1,.6); 
%\draw [ultra thick] (3.5,.2) -- (3.9,.6); 

% +1 below u
%\node [circle, draw, inner sep=1pt] at (-3.1,-4.6) {\tiny$+1$};
% +1 below v
%\node [circle, draw, inner sep=1pt] at (3.1,-4.6) {\tiny$+1$};
% -1
%\node [circle, draw, inner sep=1pt] at (2.45,-1.45) {$-1$};
% -1
%\node [circle, draw, inner sep=1pt] at (0,.65) {\tiny $-1$};
% +1 below w
%\node [circle, draw, inner sep=1pt] at (5,-1.6) {\tiny$+1$};

\end{tikzpicture}
\caption{\resolvepair$(u',v')$}
\end{subfigure}
\qquad
\begin{subfigure}[t]{.26\textwidth}
\begin{tikzpicture}[scale=.45]
\draw (-1,-1) -- (0,0) -- (3,-3);
\draw[dashed] (-2,-2) -- (-1,-1);
\draw (-2.6,-3) -- (-1.95,-1.95);
\draw[dashed] (-2,-2) -- (-1.3,-3);
\drawnode{(-2.6,-3)}
\node [below] at (-2.6,-3) {$u'$};
\drawnode{(3,-3)}
\node [below] at (3,-2.87) {$v' = \hat b$};
\drawnode{(-1.3,-3)}
\node [below] at (-1.3,-3) {$\hat r_1$};

%\node [below] at (-3,-3) {$\hat u'$};
%\drawnode{(3,-3)}
%\node [below] at (3,-2.83) {$\hat v'$};
%\drawnode{(0,0)}

\smallarbitrarytreeat{(-.6,-1.3)}
%\drawnode{(-.6,-1.3)}
%\smallarbitrarytreeat{(-1.6,-2.3)}
%\drawnode{(-1.6,-2.3)}
%\draw (-1.6,-2.3) -- (-1.95,-1.95);
\draw (-.6,-1.3) -- (-.95,-.95);
%\drawnode{(-1.95,-1.95)}
%\drawnode{(-.95,-.95)}

\smallarbitrarytreeat{(.6,-1.3)}
%\drawnode{(-.6,-1.3)}
\smallarbitrarytreeat{(1.6,-2.3)}
%\drawnode{(-1.6,-2.3)}
\draw (1.6,-2.3) -- (1.95,-1.95);
\draw (.6,-1.3) -- (.95,-.95);
%\drawnode{(-.95,-.95)}

\smallarbitrarytreeat{(2.5,1.3)}
\draw (0,0) -- (2.2,0.7);
\drawnode{(2.2,0.7)}
\drawnode{(2.8,.7)}
\draw (5,0) -- (2.8,.7);
\drawnode{(5,0)}
\node [below] at (5,0) {$\hat r_2$};

%\smallarbitrarytreeat{(5,2)}
%\draw (0,0) -- (4.7,1.4);
%\drawnode{(4.7,1.4)}
%\drawnode{(5.3,1.4)}
%\draw (10,-3) -- (5.3,1.4);
%\drawnode{(10,-3)}
%\node [below] at (10,-3) {$\hat r_2$};

% cuts and duals

% cut u
\draw [ultra thick] (2.4,-2.8) -- (2.8,-2.4); 
\draw [ultra thick] (1.95,-2.25) -- (1.55,-1.85); 
%\draw [ultra thick] (-1.7,-2.8) -- (-1.3,-2.4); 
%\draw [ultra thick] (-1.4,-1.8) -- (-1.8,-1.4); 
% cut v
%\draw [ultra thick] (2.4,-2.8) -- (2.8,-2.4); 
% cut w
%\draw [ultra thick] (9.4,-2.8) -- (9.8,-2.4); 
% cut tree
%\draw [ultra thick] (2,-2.35) -- (1.6,-1.95); 
% cut tree
%\draw [ultra thick] (1.5,.2) -- (1.1,.6); 
%\draw [ultra thick] (3.5,.2) -- (3.9,.6); 

% +1 below u
%\node [circle, draw, inner sep=1pt] at (-3.1,-4.6) {\tiny$+1$};
% +1 below v
%\node [circle, draw, inner sep=1pt] at (3.1,-4.6) {\tiny$+1$};
% -1
%\node [circle, draw, inner sep=1pt] at (2.45,-1.45) {$-1$};
% -1
%\node [circle, draw, inner sep=1pt] at (0,.65) {\tiny $-1$};
% +1 below w
%\node [circle, draw, inner sep=1pt] at (5,-1.6) {\tiny$+1$};

\end{tikzpicture}
\caption{\resolvepair$(\hat b, \hat r_2)$}
\end{subfigure}
\qquad
\begin{subfigure}[t]{.29\textwidth}
\begin{tikzpicture}[scale=.45]
\draw (-1,-1) -- (0,0) -- (2,-2);
\draw[dashed] (-2,-2) -- (-1,-1);
\draw[dashed] (2,-2) -- (3,-3);
\draw (-2.6,-3) -- (-1.95,-1.95);
\draw (-2,-2) -- (-1.3,-3);
\drawnode{(-2.6,-3)}
\node [below] at (-2.6,-3) {$u'$};
\drawnode{(3,-3)}
\node [below] at (3,-2.87) {$v' = \hat b$};
\drawnode{(-1.3,-3)}
\node [below] at (-1.3,-3) {$\hat r_1$};

%\node [below] at (-3,-3) {$\hat u'$};
%\drawnode{(3,-3)}
%\node [below] at (3,-2.83) {$\hat v'$};
%\drawnode{(0,0)}

\smallarbitrarytreeat{(-.6,-1.3)}
%\drawnode{(-.6,-1.3)}
%\smallarbitrarytreeat{(-1.6,-2.3)}
%\drawnode{(-1.6,-2.3)}
%\draw (-1.6,-2.3) -- (-1.95,-1.95);
\draw (-.6,-1.3) -- (-.95,-.95);
%\drawnode{(-1.95,-1.95)}
%\drawnode{(-.95,-.95)}

\smallarbitrarytreeat{(.6,-1.3)}
%\drawnode{(-.6,-1.3)}
\smallarbitrarytreeat{(1.6,-2.3)}
%\drawnode{(-1.6,-2.3)}
%\draw[dashed] (1.6,-2.3) -- (1.95,-1.95);
\draw (.6,-1.3) -- (.95,-.95);
%\drawnode{(-.95,-.95)}

\smallarbitrarytreeat{(2.5,1.3)}
\draw (0,0) -- (2.2,0.7);
\drawnode{(2.2,0.7)}
\drawnode{(2.8,.7)}
\draw (5,0) -- (2.8,.7);
\drawnode{(5,0)}
\node [below] at (5,0) {$\hat r_2$};

%\smallarbitrarytreeat{(5,2)}
%\draw (0,0) -- (4.7,1.4);
%\drawnode{(4.7,1.4)}
%\drawnode{(5.3,1.4)}
%\draw (10,-3) -- (5.3,1.4);
%\drawnode{(10,-3)}
%\node [below] at (10,-3) {$\hat r_2$};

% cuts and duals

% cut u
%\draw [ultra thick] (2.4,-2.8) -- (2.8,-2.4); 
%\draw [ultra thick] (1.95,-2.25) -- (1.55,-1.85); 
%\draw [ultra thick] (-1.7,-2.8) -- (-1.3,-2.4); 
%\draw [ultra thick] (-1.4,-1.8) -- (-1.8,-1.4); 
% cut v
%\draw [ultra thick] (2.4,-2.8) -- (2.8,-2.4); 
% cut w
%\draw [ultra thick] (9.4,-2.8) -- (9.8,-2.4); 
% cut tree
%\draw [ultra thick] (2,-2.35) -- (1.6,-1.95); 
% cut tree
%\draw [ultra thick] (1.5,.2) -- (1.1,.6); 
%\draw [ultra thick] (3.5,.2) -- (3.9,.6); 

% +1 below u
%\node [circle, draw, inner sep=1pt] at (-3.1,-4.6) {\tiny$+1$};
% +1 below v
%\node [circle, draw, inner sep=1pt] at (3.1,-4.6) {\tiny$+1$};
% -1
%\node [circle, draw, inner sep=1pt] at (2.45,-1.45) {$-1$};
% -1
%\node [circle, draw, inner sep=1pt] at (0,.65) {\tiny $-1$};
% +1 below w
%\node [circle, draw, inner sep=1pt] at (5,-1.6) {\tiny$+1$};

\draw [ultra thick] (3.5,.2) -- (3.9,.6);

\end{tikzpicture}
\caption{Cut off $\hat r_2$, and retroactively merge $u'$ and $W=\{\hat r_1\}$.}
\end{subfigure}
\end{center}
\caption{Illustration of a case where a retroactive merge is needed. The set $R$ contains leaves $\hat r_1,\hat r_2$ only, and the set $B$ contains two leaves, $u'$ and $v'$ (where $v'$ will be the last remaining active leaf in $B$, so $v'=\hat b$).
Figure (a) shows the execution of \resolvepair$(u',v')$. After this, we execute Procedure~\ref{fig:twotrees}, with $\hat u=\hat r_1$, $\{\hat v_1,\hat v_2\}= \{\hat b,\hat r_2\}$.}\label{fig:retrofig}
\end{figure}

\begin{figure}
\begin{center}
\begin{subfigure}[t]{.26\textwidth}
\begin{tikzpicture}[scale=.45]
\draw (-2,-2) -- (0,0) -- (3,-3);
\draw (-2.6,-3) -- (-1.95,-1.95) -- (-1.3,-3);
\drawnode{(-2.6,-3)}
\node [below] at (-2.6,-3) {$u'$};
\drawnode{(3,-3)}
\node [below] at (3,-3) {$v' = \hat r_1$};
\drawnode{(-1.3,-3)}
\node [below] at (-1.3,-2.87) {$\hat b$};

%\node [below] at (-3,-3) {$\hat u'$};
%\drawnode{(3,-3)}
%\node [below] at (3,-2.83) {$\hat v'$};
%\drawnode{(0,0)}

\smallarbitrarytreeat{(-.6,-1.3)}
%\drawnode{(-.6,-1.3)}
%\smallarbitrarytreeat{(-1.6,-2.3)}
%\drawnode{(-1.6,-2.3)}
%\draw (-1.6,-2.3) -- (-1.95,-1.95);
\draw (-.6,-1.3) -- (-.95,-.95);
%\drawnode{(-1.95,-1.95)}
%\drawnode{(-.95,-.95)}

\smallarbitrarytreeat{(.6,-1.3)}
%\drawnode{(-.6,-1.3)}
\smallarbitrarytreeat{(1.6,-2.3)}
%\drawnode{(-1.6,-2.3)}
\draw (1.6,-2.3) -- (1.95,-1.95);
\draw (.6,-1.3) -- (.95,-.95);
%\drawnode{(-.95,-.95)}

\smallarbitrarytreeat{(2.5,1.3)}
\draw (0,0) -- (2.2,0.7);
\drawnode{(2.2,0.7)}
\drawnode{(2.8,.7)}
\draw (5,0) -- (2.8,.7);
\drawnode{(5,0)}
\node [below] at (5,0) {$\hat r_2$};

%\smallarbitrarytreeat{(5,2)}
%\draw (0,0) -- (4.7,1.4);
%\drawnode{(4.7,1.4)}
%\drawnode{(5.3,1.4)}
%\draw (10,-3) -- (5.3,1.4);
%\drawnode{(10,-3)}
%\node [below] at (10,-3) {$\hat r_2$};

% cuts and duals

% cut u
%\draw [ultra thick] (-2.2,-2.8) -- (-2.6,-2.4); 
\draw [ultra thick] (-1.7,-2.8) -- (-1.3,-2.4); 
\draw [ultra thick] (-1.4,-1.8) -- (-1.8,-1.4); 
% cut v
%\draw [ultra thick] (2.4,-2.8) -- (2.8,-2.4); 
% cut w
%\draw [ultra thick] (9.4,-2.8) -- (9.8,-2.4); 
% cut tree
%\draw [ultra thick] (2,-2.35) -- (1.6,-1.95); 
% cut tree
%\draw [ultra thick] (1.5,.2) -- (1.1,.6); 
%\draw [ultra thick] (3.5,.2) -- (3.9,.6); 

% +1 below u
%\node [circle, draw, inner sep=1pt] at (-3.1,-4.6) {\tiny$+1$};
% +1 below v
%\node [circle, draw, inner sep=1pt] at (3.1,-4.6) {\tiny$+1$};
% -1
%\node [circle, draw, inner sep=1pt] at (2.45,-1.45) {$-1$};
% -1
%\node [circle, draw, inner sep=1pt] at (0,.65) {\tiny $-1$};
% +1 below w
%\node [circle, draw, inner sep=1pt] at (5,-1.6) {\tiny$+1$};

\end{tikzpicture}
\caption{\resolvepair$(u',v')$}
\end{subfigure}
\qquad
\begin{subfigure}[t]{.26\textwidth}
\begin{tikzpicture}[scale=.45]
\draw (-1,-1) -- (0,0) -- (3,-3);
\draw[dashed] (-2,-2) -- (-1,-1);
\draw (-2.6,-3) -- (-1.95,-1.95);
\draw[dashed] (-2,-2) -- (-1.3,-3);
\drawnode{(-2.6,-3)}
\node [below] at (-2.6,-3) {$u'$};
\drawnode{(3,-3)}
\node [below] at (3,-3) {$v' = \hat r_1$};
\drawnode{(-1.3,-3)}
\node [below] at (-1.3,-2.87) {$\hat b$};

%\node [below] at (-3,-3) {$\hat u'$};
%\drawnode{(3,-3)}
%\node [below] at (3,-2.83) {$\hat v'$};
%\drawnode{(0,0)}

\smallarbitrarytreeat{(-.6,-1.3)}
%\drawnode{(-.6,-1.3)}
%\smallarbitrarytreeat{(-1.6,-2.3)}
%\drawnode{(-1.6,-2.3)}
%\draw (-1.6,-2.3) -- (-1.95,-1.95);
\draw (-.6,-1.3) -- (-.95,-.95);
%\drawnode{(-1.95,-1.95)}
%\drawnode{(-.95,-.95)}

\smallarbitrarytreeat{(.6,-1.3)}
%\drawnode{(-.6,-1.3)}
\smallarbitrarytreeat{(1.6,-2.3)}
%\drawnode{(-1.6,-2.3)}
\draw (1.6,-2.3) -- (1.95,-1.95);
\draw (.6,-1.3) -- (.95,-.95);
%\drawnode{(-.95,-.95)}

\smallarbitrarytreeat{(2.5,1.3)}
\draw (0,0) -- (2.2,0.7);
\drawnode{(2.2,0.7)}
\drawnode{(2.8,.7)}
\draw (5,0) -- (2.8,.7);
\drawnode{(5,0)}
\node [below] at (5,0) {$\hat r_2$};

%\smallarbitrarytreeat{(5,2)}
%\draw (0,0) -- (4.7,1.4);
%\drawnode{(4.7,1.4)}
%\drawnode{(5.3,1.4)}
%\draw (10,-3) -- (5.3,1.4);
%\drawnode{(10,-3)}
%\node [below] at (10,-3) {$\hat r_2$};

% cuts and duals

% cut u
\draw [ultra thick] (2.4,-2.8) -- (2.8,-2.4); 
\draw [ultra thick] (1.95,-2.25) -- (1.55,-1.85); 
%\draw [ultra thick] (-1.7,-2.8) -- (-1.3,-2.4); 
%\draw [ultra thick] (-1.4,-1.8) -- (-1.8,-1.4); 
% cut v
%\draw [ultra thick] (2.4,-2.8) -- (2.8,-2.4); 
% cut w
%\draw [ultra thick] (9.4,-2.8) -- (9.8,-2.4); 
% cut tree
%\draw [ultra thick] (2,-2.35) -- (1.6,-1.95); 
% cut tree
%\draw [ultra thick] (1.5,.2) -- (1.1,.6); 
%\draw [ultra thick] (3.5,.2) -- (3.9,.6); 

% +1 below u
%\node [circle, draw, inner sep=1pt] at (-3.1,-4.6) {\tiny$+1$};
% +1 below v
%\node [circle, draw, inner sep=1pt] at (3.1,-4.6) {\tiny$+1$};
% -1
%\node [circle, draw, inner sep=1pt] at (2.45,-1.45) {$-1$};
% -1
%\node [circle, draw, inner sep=1pt] at (0,.65) {\tiny $-1$};
% +1 below w
%\node [circle, draw, inner sep=1pt] at (5,-1.6) {\tiny$+1$};

\end{tikzpicture}
\caption{\resolvepair$(\hat r_1, \hat r_2)$}
\end{subfigure}
\qquad
\begin{subfigure}[t]{.29\textwidth}
\begin{tikzpicture}[scale=.45]
\draw (-1,-1) -- (0,0) -- (2,-2);
\draw[dashed] (-2,-2) -- (-1,-1);
\draw[dashed] (2,-2) -- (3,-3);
\draw (-2.6,-3) -- (-1.95,-1.95);
\draw (-2,-2) -- (-1.3,-3);
\drawnode{(-2.6,-3)}
\node [below] at (-2.6,-3) {$u'$};
\drawnode{(3,-3)}
\node [below] at (3,-3) {$v' = \hat r_1$};
\drawnode{(-1.3,-3)}
\node [below] at (-1.3,-2.87) {$\hat b$};

%\node [below] at (-3,-3) {$\hat u'$};
%\drawnode{(3,-3)}
%\node [below] at (3,-2.83) {$\hat v'$};
%\drawnode{(0,0)}

\smallarbitrarytreeat{(-.6,-1.3)}
%\drawnode{(-.6,-1.3)}
%\smallarbitrarytreeat{(-1.6,-2.3)}
%\drawnode{(-1.6,-2.3)}
%\draw (-1.6,-2.3) -- (-1.95,-1.95);
\draw (-.6,-1.3) -- (-.95,-.95);
%\drawnode{(-1.95,-1.95)}
%\drawnode{(-.95,-.95)}

\smallarbitrarytreeat{(.6,-1.3)}
%\drawnode{(-.6,-1.3)}
\smallarbitrarytreeat{(1.6,-2.3)}
%\drawnode{(-1.6,-2.3)}
%\draw[dashed] (1.6,-2.3) -- (1.95,-1.95);
\draw (.6,-1.3) -- (.95,-.95);
%\drawnode{(-.95,-.95)}

\smallarbitrarytreeat{(2.5,1.3)}
\draw (0,0) -- (2.2,0.7);
\drawnode{(2.2,0.7)}
\drawnode{(2.8,.7)}
\draw (5,0) -- (2.8,.7);
\drawnode{(5,0)}
\node [below] at (5,0) {$\hat r_2$};

%\smallarbitrarytreeat{(5,2)}
%\draw (0,0) -- (4.7,1.4);
%\drawnode{(4.7,1.4)}
%\drawnode{(5.3,1.4)}
%\draw (10,-3) -- (5.3,1.4);
%\drawnode{(10,-3)}
%\node [below] at (10,-3) {$\hat r_2$};

% cuts and duals

% cut u
%\draw [ultra thick] (2.4,-2.8) -- (2.8,-2.4); 
%\draw [ultra thick] (1.95,-2.25) -- (1.55,-1.85); 
%\draw [ultra thick] (-1.7,-2.8) -- (-1.3,-2.4); 
%\draw [ultra thick] (-1.4,-1.8) -- (-1.8,-1.4); 
% cut v
%\draw [ultra thick] (2.4,-2.8) -- (2.8,-2.4); 
% cut w
%\draw [ultra thick] (9.4,-2.8) -- (9.8,-2.4); 
% cut tree
%\draw [ultra thick] (2,-2.35) -- (1.6,-1.95); 
% cut tree
%\draw [ultra thick] (1.5,.2) -- (1.1,.6); 
%\draw [ultra thick] (3.5,.2) -- (3.9,.6); 

% +1 below u
%\node [circle, draw, inner sep=1pt] at (-3.1,-4.6) {\tiny$+1$};
% +1 below v
%\node [circle, draw, inner sep=1pt] at (3.1,-4.6) {\tiny$+1$};
% -1
%\node [circle, draw, inner sep=1pt] at (2.45,-1.45) {$-1$};
% -1
%\node [circle, draw, inner sep=1pt] at (0,.65) {\tiny $-1$};
% +1 below w
%\node [circle, draw, inner sep=1pt] at (5,-1.6) {\tiny$+1$};

\draw [ultra thick] (3.5,.2) -- (3.9,.6);

\end{tikzpicture}
\caption{Cut off $\hat r_2$, and retroactively merge $u'$ and $W=\{\hat b\}$.}
\end{subfigure}
\end{center}
\caption{Illustration of a second case where a retroactive merge is needed. The set $R$ contains three leaves $u',\hat r_1,\hat r_2$, and the set $B$ contains a single leaf, $\hat b$.
Figure (a) shows the execution of \resolvepair$(u',v')$, with $v'=\hat r_1$. After this, we execute Procedure~\ref{fig:twotrees}, with $\hat u=\hat b$, $\{\hat v_1,\hat v_2\}= \{\hat r_1,\hat r_2\}$.}\label{fig:retrofig2}
\end{figure}
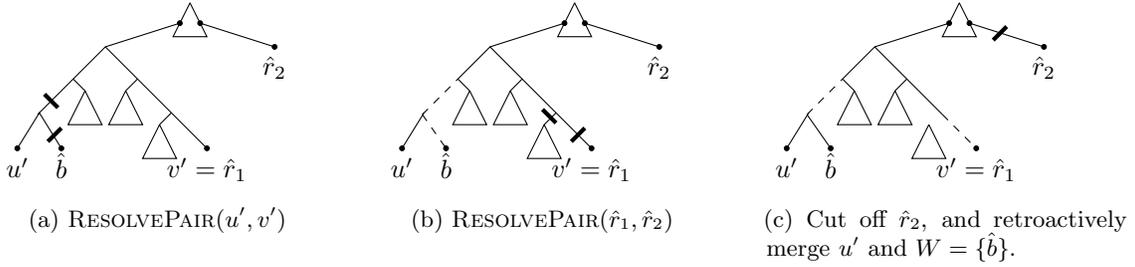

\begin{procedure}
\caption{4b(): $\hat r_1,\hat r_2$ and $\hat b$ are in three distinct trees in $T_2'$}\label{fig:threetrees}
Cut off $\hat b$ in $T_2'$ and $T_1'$ and make $\hat b$ inactive. \dual{$y_{\hat b}\ot 1$.} See Claim~\ref{realcut}.\;\label{alg:cutb}
Cut off $\hat r_1$ in $T_2'$ (if $\hat r_1$'s tree contains at least one other active leaf) and in $T_1'$ and make $\hat r_1$ inactive.	\dual{$y_{\hat r_1}\ot 1$.}\;\label{alg:prev} 
Cut off $\hat r_2$ in $T_2'$ (if $\hat r_2$'s tree contains at least one other active leaf) and in $T_1'$ and make $\hat r_2$ inactive.	\dual{$y_{\hat r_2}\ot 1$.}\;\label{alg:prev2} 
\If { ($\hat r_1$ and $\hat r_2$ were each the only active leaf in its tree in $T_2'$ in lines~\ref{alg:prev}--\ref{alg:prev2}) {\em and} (no {\em FinalCut} or {\em Merge-After-Cut} was executed in some call to \resolvepair\ in lines~\ref{alg:beginresolveB}--\ref{alg:endresolveB}) \label{alg:ifretro1}}
{
	Find the inactive leaves $u',v'\in B$ for which the execution of \resolvepair$(u',v')$ cut off set $W\ni\hat r_2$ in line~\ref{alg:2cutW}. Let $u'$ be the leaf that was deactivated by \resolvepair$(u',v')$. \label{alg:retromerge1}
	Retroactively merge the trees containing $W$ and $u'$ (see Claims~\ref{claim:retro} and~\ref{claim:freenodes}, and the remarks following these).\; 
}	
 \end{procedure}
 \begin{procedure}
\caption{4c(): $\hat r_1,\hat r_2$ and $\hat b$ are in two distinct trees in $T_2'$}\label{fig:twotrees}
Relabel $\hat r_1,\hat r_2,\hat b$ as $\hat u,\hat v_1,\hat v_2$, so that $\hat u$ is the only active leaf in $R\cup B$ in its tree in $T_2'$, and $\hat v_1,\hat v_2$ are in the same tree in  $T_2'$.\;
\eIf {$\hat u$ is the only active leaf in its tree in $T_2'$}
{
	Cut off $\hat u$ in $T_1'$ and make $u$ inactive. \dual{$y_{\hat u}\ot 1$.}\;	\label{alg:alone}
}
{
	Cut off $\hat u$ in $T_2'$ and $T_1'$ and make $u$ inactive. \dual{$y_{\hat u}\ot 1$.}\;\label{alg:notalone}
}	
\eIf{$\hat v_1,\hat v_2$ are active siblings in $T_2'$\label{alg:ifsib}}
{
	Relabel $\hat v_1$ and $\hat v_2$ if necessary so $\hat v_1=\hat b$. See Claim~\ref{claim:rbmerge}.\;
	Merge $\hat v_1$ and $\hat v_2$ (i.e., make $\hat v_1$ inactive to ``merge'' it with $\hat v_2$). \dual{$y_{\hat b}\ot 1$.} .\;\label{alg:rbmerge}
}
{	
	\resolvepair$(\hat v_1,\hat v_2)$.\;\label{alg:lastpair}
	Cut off the last active leaf in $R\cup B$, say $\hat v_2$, in $T_2'$ and $T_1'$ and make $\hat v_2$ inactive. \dual{$y_{\hat v_2}\ot 1$.} See Claim~\ref{realcut}.\;\label{alg:cutp}
	\If { (line~\ref{alg:alone} was executed) {\em{and}} (no {\em FinalCut} or {\em Merge-After-Cut} was executed in line~\ref{alg:lastpair} or lines~\ref{alg:beginresolveB}--\ref{alg:endresolveB}) \label{alg:ifretro2}}
	{
		Find the inactive leaves $u',v'\in R\cup B$ for which the execution of \resolvepair$(u',v')$ cut off set $W\ni \hat u$ in line~\ref{alg:2cutW}. Let $u'$ be the leaf that was deactivated by \resolvepair$(u',v')$.\label{alg:retromerge2}
		Retroactively merge $W$ and $u'$ (see Claims~\ref{claim:retro} and~\ref{claim:freenodes}, and the remarks following these).\;
	}	
}
\end{procedure}

% EXPLANATION OF RETROACTIVE MERGE AS WELL AS PROOF OF APPROXIMATION RATIO FOR THIS CASE

The analysis for the two procedures overlaps in part. In the current subsection, we will give these parts of the analysis, including an explanation of what we mean by retroactively merging $W$ and $u'$ in the last line of the two procedures. The proofs in this section are quite technical, and to maintain the flow of the argument, they have been deferred to Section~\ref{sec:defer}.

In Section~\ref{sec:threetrees} and Section~\ref{sec:twotrees}, we show that a pass through the while-loop that goes through Procedure~\ref{fig:threetrees} and Procedure~\ref{fig:twotrees}, respectively, outputs a valid tuple, and increases the dual objective value by at least half of the number of edges cut from $T_2'$.

Our first claim implies that in line~\ref{alg:cutb} and line~\ref{alg:cutp} the dual objective value will increase by $1$.
\begin{claim}\label{realcut}
Just before line~\ref{alg:cutb} and line~\ref{alg:cutp}, the leaf $\hat b$ and $\hat v_2$, respectively, is not the only active leaf in its tree in $T_2'$. 
\end{claim}

The difficulty in Procedures~\ref{fig:threetrees} and~\ref{fig:twotrees} lies in the case where the condition in line~\ref{alg:ifretro1} or line~\ref{alg:ifretro2}, respectively, evaluates to true. In this case, the dual objective in the current pass through the while-loop has not increased enough to ``pay for'' (half of) the edges that were deleted from $T_2'$. 

Luckily, it turns out that will be able to ``retroactively merge'' two inactive trees in the two forests in this case. %We will refer to this as a ``retroactive'' merge. 
The next claim allows us to identify two inactive trees, one for which the leaves that were active at the start of the current pass through the while-loop are in $R$ and one for which the leaves that were active at the start of the while-loop are in $B$. Furthermore, we will be able to identify a single call to \resolvepair\ in the algorithm in which these two trees were both cut off from $T_2'$. 

In the following claim, we say that two leaves have the same color, if they are in the same set, either $R$ or $B$.
\begin{claim}\label{claim:retro}
The leaves $u',v'$ in line~\ref{alg:retromerge1} and line~\ref{alg:retromerge2} exist. Furthermore, the set $W$ in line~\ref{alg:retromerge1} and line~\ref{alg:retromerge2} contains only active leaves that have the same color as $\hat r_2$ and $\hat u$, respectively.
\end{claim}
Note that either $u'\in R$ and $W\subseteq B$, or $u'\in B$ and $W\subseteq R$. In the latter case, we must have that $W$ is a singleton, since the operation of \resolvepair$(u',v')$ must have occurred after completing \resolveset$(R)$. In the former case, the active leaves in $W$ will be merged by calls to \resolvepair\ until only $\hat u$ remains.
%It thus follows from Claim~\ref{claim:retro}, that $W\cup \{u'\}$ formed a compatible active sibling set in $T_2'$ at the moment when \resolvepair$(u',v')$ was called. We were not able to merge them at that time, because they did not form an active sibling set in $T_1'$ (since the subtree rooted at $\lca_1(W\cup \{u'\})$ still contained other active leaves in $R\cup B$, for example $v'$).
Retroactively merging $W$ and $u'$ means that we want to restore the paths connecting them in both forests.
In order to retroactively merge them, we thus need to show that nodes on the paths between these two trees in $T_1'$ and $T_2'$ are not ``used'' to connect two leaves in another active or inactive tree.
This is what is shown in the proof of the following claim.
For a set of leaves $L$, we let $V_1[L]$ be the nodes in $V[L]$ that are in $T_1$, and we let $V_2[L]$ be the nodes in $V[L]$ that are in $T_2$.
\begin{claim}\label{claim:freenodes}
Let $u'$ be the leaf and $W$ the set identified in line~\ref{alg:retromerge1} or line~\ref{alg:retromerge2}, respectively, and let $L_1, \ldots, L_k$ be the leaf sets of the trees in $T_1'$ at this moment, excluding the trees containing $u'$ and $W$, and let $M_1, \ldots, M_{k'}$  be the leaf sets of the trees in $T_2'$ at this moment, excluding the trees containing $u'$ and $W$. 
Then, $V_1[L_j]\cap V_1[\{u'\}\cup W]=\emptyset$ for all $j=1,\ldots, k$ and $V_2[M_{j'}]\cap V_2[\{u'\}\cup W]=\emptyset$ for all $j'=1,\ldots, k'$.
\end{claim}

By the claim, we can merge the trees containing $u'$ and $W$: we reinsert the missing edges connecting the leaves in $V[\{u'\}\cup W]$ in the two forests, each time adding an edge with one endpoint in the tree containing $W$ and the other endpoint in $V[\{u'\}\cup W]$ but not in the tree containing $W$.
Whenever the addition (undeletion) of an edge merges the tree containing $W$ with $L_j$ or $M_{j'}$ for some $j$ or $j'$, there must be a node in $V[\{u'\}\cup W]$ that is incident to an edge, for which the other endpoint is not in $V[\{u'\}\cup W]$. We delete the latter edge, and continue. By the claim, this edge exists, and deleting this edge does not disconnect $L_j$ or $M_{j'}$.
   
Since the number of trees in both forests decreases by $1$ (since the edge that connects the tree containing $W$ and the tree containing $u'$ does not lead to a deleted edge), we have thus shown that the retroactive merge decreases $|E(T_2)\setminus E(T_2')|$ by $1$.

% Let $u'$ be the leaf and $W$ the set identified in line~\ref{alg:retromerge1}, respectively line~\ref{alg:retromerge2}. 
% Let $S(u')$ and $S(W)$ be the set of inactive leaves that are in the same tree as $u'$ and $W$. Let $E_i$ be the edges on the path in the original tree $T_i$ from $\lca_i(S(u'))$ to $\lca_i(S(W))$ for $i=1,2$.
% Then in the current forest $T_i'$ no edge in $E_i$ is on the path between leaves in ${\cal L}$.

\subsubsection{Analysis of Procedure~\ref{fig:threetrees}}\label{sec:threetrees}

\begin{lemma}\label{lemma:threetrees}
Let $(T_1',T_2',{\cal L}',y)$ be a valid tuple, that has been preprocessed by procedure \preprocess, and 
let $R\cup B$ be a minimal incompatible active sibling set with $\lca_2(R)=\lca_2(R\cup B)$, for which \resolveset$( R )$ returns {\sc Fail}.
If, after executing lines~\ref{alg:beginwhile}--\ref{alg:endresolveB}, the three remaining active leaves are in three distinct trees in $T_2'$, then the tuple $(\tilde T_1',\tilde T_2',\tilde {\cal L}',\tilde y)$ 
after executing lines~\ref{alg:beginwhile}--\ref{alg:endresolveB} followed by Procedure~\ref{fig:threetrees} is valid, and satisfies
\[D(\tilde T_2',\tilde{\cal L}',\tilde y)-D(T_2',{\cal L}',y) \ge \tfrac 12 \big(|E(T_2')\setminus E(\tilde T_2')|\big).\]
\end{lemma}
\begin{proof}
It follows from the remarks about the retroactive merge at the end of the previous subsection that the Procedure~\ref{fig:threetrees} maintains all properties of a valid tuple, except possibly for the feasibility of the modified dual solution. In addition, we need to show that the increase in the dual objective value can ``pay for'' half of the increase in the primal objective value.

As before, we know by Lemma~\ref{lemma:safe2} that the tuple is valid and $\{\hat r_1,\hat r_2\}$-safe and $\{\hat b\}$-safe at the start of Procedure~\ref{fig:threetrees}. We go through the lines of Procedure~\ref{fig:threetrees} and consider the effect on the dual solution.
\begin{itemize}
\item[\ref{alg:cutb}.]
By Claim~\ref{realcut}, the active leaf set $A$ of the tree in $T_2'$ containing $\hat b$ before executing line~\ref{alg:cutb} contains at least one leaf in addition to $\hat b$. Executing this line, decreases $z_{A}$ by $1$ and increases $z_{A\setminus \{\hat b\}}$ and $y_{\hat b}$ by $1$. This increases the load on compatible sets $L$ containing $\hat b$ and some leaf in $A\setminus\{\hat b\}$, and by $\{\hat b\}$-safeness, the load on these sets was at most 0. Hence, the dual solution remains feasible. The objective value of the dual solution increases by $1$.
\item[\ref{alg:prev}--\ref{alg:prev2}.] 
Let $A_1$ and $A_2$ be the active leaves of the tree in $T_2'$ containing $\hat r_1$ and $\hat r_2$, respectively, before executing lines~\ref{alg:prev}--\ref{alg:prev2}. 
The effect of these lines on the dual solution depends on whether $A_1,A_2$ contain other leaves or not. If they do not, then the dual solution is effectively unchanged, since if $A_i=\{\hat r_i\}$, we simply decrease $z_{A_i}$ by $1$ and we increase $y_{\hat r_i}$ by $1$.
If $A_i$ contains some leaf in addition to $r_i$, we also increase $z_{A_i\setminus\{\hat r_i\}}$ by $1$.

Note that the load on a compatible set $L$ increases by these lines, if $L$ contains $\hat r_i$ and a leaf in $A_i\setminus\{\hat r_i\}$ for $i$ equal to 1 and/or 2.
By the fact the dual solution was $\{\hat r_1,\hat r_2\}$-safe at the start of the procedure, we know that the load on $L$ was at most 0 at the start of the procedure.

Furthermore, we claim that a compatible set $L$ can get an increase in its load from only one of the three ``splits'', i.e., either because it contains $\hat b$ and a leaf in $A\setminus \{\hat b\}$, or because it contains $\hat r_1$ and a leaf in $A_1\setminus \{\hat r_1\}$, or because it contains $\hat r_2$ and a leaf in $A_1\setminus \{\hat r_2\}$. 
This is because at most one the trees in $T_2'$ at the start of the procedure contains $\lca_2(R\cup B)$, and hence at most one of $A\setminus\{\hat b\}, A_1\setminus \{\hat r_1\}$ and $A_2\setminus\{\hat r_2\}$ contains leaves $x$ such that $uv|x$ for $u,v\in R\cup B$. If $L$ contains a leaf $x$ such that $uv|x$ does not hold for $u,v\in R\cup B$, then it cannot contain two leaves in $R\cup B$.

Thus, the load increases by at most 1 on any compatible set $L$, and a set for which the load increases had load 0 at the start of the procedure.
The remaining lines do not affect the dual solution, and hence, we have shown that the dual solution remains feasible.
\end{itemize}
We now consider the total change in the primal and dual objective value.
Let $\Delta P_1$ be the number of edges in $E(T_2')\setminus E(\tilde T_2')$  due to lines~\ref{alg:beginwhile}--\ref{alg:endresolveB}, and let $\Delta P_2$ be the number of edges in $E(T_2')\setminus E(\tilde T_2')$  due to lines~\ref{alg:cutb}--\ref{alg:prev2} of Procedure~\ref{fig:threetrees}. 
Similarly, let $\Delta D_1$ be the total change in the dual objective value by lines~\ref{alg:beginwhile}--\ref{alg:endresolveB}, and
$\Delta D_2$ the change in the dual objective due to Procedure~\ref{fig:threetrees}.

As in the proof of Lemma~\ref{lemma:onetree}, we have $\Delta D_1 \ge \tfrac 12 \Delta P_1 -1$. Furthermore, we know that if at least one {\it FinalCut} or {\it Merge-After-Cut} is executed, that $\Delta D_1\ge \tfrac 12\Delta P_1-1 + \tfrac 12$.

Each edge deleted from $T_2'$ by lines~\ref{alg:cutb}--\ref{alg:prev2} contribute the same amount to $\Delta P_2$ as to $\Delta D_2$, so $\Delta P_2=\Delta D_2$, or, equivalently, $\Delta D_2=\tfrac12 \Delta P_2 + \tfrac12 \Delta P_2$.
Letting $c=1$ if at least one {\it FinalCut} or {\it Merge-After-Cut} is executed and 0 otherwise, we thus have that 
$\Delta D_1+\Delta D_2\ge \tfrac 12\left(\Delta P_1+\Delta P_2\right) -1 + \tfrac 12c+\tfrac12\Delta P_2$.

By Claim~\ref{realcut}, we know that $\Delta P_2\ge 1$. Therefore, if $c=1$, or if $\Delta P_2\ge 2$, then 
$\Delta D_1+\Delta D_2\ge \tfrac 12\left(\Delta P_1+\Delta P_2\right)$. 
If neither of those holds, then the retroactive merge ensures that $|E(T'_2)\setminus E(\tilde T_2')|=\Delta P_1+\Delta P_2-1$, and so we again have
$\Delta D_1+\Delta D_2\ge \tfrac 12|E(T'_2)\setminus E(\tilde T_2')|$. 
\end{proof}

\subsubsection{Analysis of Procedure~\ref{fig:twotrees}}\label{sec:twotrees}

The following claim will be needed to show that if $\hat v_1, \hat v_2$ are active siblings in $T_2'$, then the dual solution remains feasible when we set $y_{\hat v_1}$ to $1$. The proof is deferred to Section~\ref{sec:defer}.
\begin{claim}\label{claim:rbmerge}
If $\hat v_1$ and $\hat v_2$ are active siblings in $T_2'$ in line~\ref{alg:ifsib} of Procedure~\ref{fig:twotrees}, then one of them is $\hat b$.
\end{claim}

\begin{lemma}\label{lemma:twotrees}
Let $(T_1',T_2',{\cal L}',y)$ be a valid tuple, that has been preprocessed by procedure \preprocess, and 
let $R\cup B$ be a minimal incompatible active sibling set with $\lca_2(R)=\lca_2(R\cup B)$, for which \resolveset$( R )$ returns {\sc Fail}.
If, after executing lines~\ref{alg:beginwhile}--\ref{alg:endresolveB}, the three remaining active leaves are in two distinct trees in $T_2'$, then the tuple $(\tilde T_1',\tilde T_2',\tilde {\cal L}',\tilde y)$ 
after executing lines~\ref{alg:beginwhile}--\ref{alg:endresolveB} followed by Procedure~\ref{fig:twotrees} is valid, and satisfies
\[D(\tilde T_2',\tilde{\cal L}',\tilde y)-D(T_2',{\cal L}',y) \ge \tfrac 12 \big(|E(T_2')\setminus E(\tilde T_2')|\big).\]
\end{lemma}
\begin{proof}
We first argue that the tuple after executing Procedure~\ref{fig:twotrees} has all the properties of a valid tuple, except possibly for dual feasibility. The only operation that is executed that we did not see before, is the handling of $\hat v_1,\hat v_2$, which may not both be in $R$ or $B$. However, we have cut off the only other active leaf in $R\cup B$, i.e., $\hat u$, in line~\ref{alg:alone} or line~\ref{alg:notalone}, and we thus know that $\hat v_1,\hat v_2$ are indeed active siblings in $T_1'$ when we merge them or execute \resolvepair$(\hat v_1,\hat v_2)$ in the next lines of the procedure.

By Lemma~\ref{lemma:safe2}, we know that the dual solution is valid and $\{\hat b\}$-safe and $\{\hat r_1,\hat r_2\}$-safe at the start of the procedure.
We now consider the dual solution, by looking at the different ways in which Procedure~\ref{fig:twotrees} may be executed.
We let $\Delta P_2$ be the number of edges in $E(T_2')\setminus E(\tilde T_2')$  due to Procedure~\ref{fig:twotrees}, and let $\Delta D_2$ the change in the dual objective due to Procedure~\ref{fig:twotrees}.
\begin{itemize}
\item
$\hat u$ is cut off, followed by a merge of active sibling pair $\hat v_1$ and $\hat v_2$.

As in lines~\ref{alg:prev}--\ref{alg:prev2} of Procedure~\ref{fig:threetrees}, the effect of cutting off $\hat u$ on the dual solution is effectively zero if $\hat u$ was the only active leaf in its tree in $T_2'$, and otherwise we have deleted one edge from $T_2'$ and we increase the dual objective by $1$. The load is increased by $1$ only on sets $L$ containing $\hat u$ and another active leaf in the tree in $T_2'$ that contained $\hat u$. Note that $\hat u\in \{\hat b,\hat r_1,\hat r_2\}$, and the dual solution thus remains feasible, by the fact that the dual solution was $\{\hat b\}$-safe and $\{\hat r_1,\hat r_2\}$-safe.

We now consider the effect of merging $\hat v_1$ and $\hat v_2$. Recall that $\hat v_1=\hat b$ because of the relabeling and Claim~\ref{claim:rbmerge}. Assume without loss of generality that $\hat v_2=\hat r_2$ and $\hat u=\hat r_1$.
If $A$ is the set of active leaves in the tree in $T_2'$ containing $\hat v_1,\hat v_2$, then the dual is changed by decreasing $z_A$ by $1$, and increasing $y_{\hat b}$ and $z_{A\setminus \{\hat b\}}$ by $1$.
This increases the dual objective by $1$, and it increases the load only on sets $L$ containing both $\hat b$ and a leaf in $A\setminus\{\hat b\}$. The load on such a set $L$ must have been at most 0 at the start of the procedure, by the fact that the dual solution was $\{\hat b\}$-safe.

Finally, we argue that the load on a compatible set cannot increase twice by the above: Suppose it did, then $L$ must contain $\hat u=\hat r_1$ and $\hat b$ and a leaf in $A\setminus \{\hat b\}$ and an active leaf in the tree in $T_2'$ that remains after cutting off $\hat u$.
At the start of the procedure, it cannot have been the case that both the tree in $T_2'$ containing $\hat v_1,\hat v_2$ and the tree in $T_2'$ containing $\hat u$ contained $\lca_2(R\cup B)$. Hence, for one of these trees the leaves $x\not\in B\cup R$ do not satisfy $\hat r_1\hat b|x$ in $T_2$, and thus, such a set $L$ is not compatible. The dual solution therefore remains feasible.

We have that $\Delta P_2$ is either 0 or 1, and $\Delta D_2=1+\Delta P_2$. 
\item
$\hat u$ is cut off, followed by \resolvepair$(\hat v_1,\hat v_2)$, after which $\hat v_1$ is cut off.

Noting as before that executing line~\ref{alg:alone} does not effectively change the dual solution, we have three operations that potentially change the dual solution and increase the load on a set $L$:
\begin{itemize}
\item If line~\ref{alg:notalone} is executed, let $A_{\hat u}$ be the active leaves of the tree in $T_2'$ containing $\hat u$ before cutting off $\hat u$.
Then $z_{A_{\hat u}}$ is decreased by $1$, and $z_{A_{\hat u}\setminus\{\hat u\}}$ and $y_{\hat u}$ are increased by $1$. The load is thus increased only on a set $L$ if it contains $\hat u$, and some leaf in $A_{\hat u}\setminus \{\hat u\}$.
The dual objective value increases by $1$.
\item
To consider the effect of the load on a set $L$ when executing \resolvepair$(\hat v_1,\hat v_2)$, note that $\hat v_1,\hat v_2$ are in the same tree in $T_2'$ and are not active siblings in $T_2'$ when executing \resolvepair$(\hat v_1,\hat v_2)$. Let $\hat v_2$ be the leaf that remains active after executing \resolvepair$(\hat v_1,\hat v_2)$, and let $A_{\hat v_1}\setminus\{\hat v_1\}$ be the set $W$ that is cut off by line~\ref{alg:2cutW} of \resolvepair. 
From the discussion in the proof of Lemma~\ref{lemma:safe}, the load increases only for a set $L$ containing $\hat v_1$ and a leaf in $A_{\hat v_1}\setminus\{\hat v_1\}$.
\item
When executing line~\ref{alg:cutp}, let $A_{\hat v_2}$ be the active leaves in the tree in $T_2'$ containing $\hat v_2$ before cutting of $\hat v_2$.
By Claim~\ref{realcut}, $A_{\hat v_2}\setminus\{\hat v_2\}\neq\emptyset$.
Then $z_{A_{\hat v_2}}$ is decreased by $1$, and $z_{A_{\hat v_2}\setminus\{\hat v_2\}}$ and $y_{\hat v_2}$ are increased by $1$. The load is thus increased on a set $L$ if it contains $\hat v_2$, and some leaf in $A_{\hat v_2}\setminus \{\hat v_2\}$.
The dual objective value increases by $1$.
\end{itemize}
By the fact that the dual solution was $\{\hat b\}$-safe and $\{\hat r_1,\hat r_2\}$-safe, the load on any set $L$ that has its load increased by one or more of the above must have had load 0 at the start of the procedure.  
Suppose there is a set $L$ that has its load increased by more than $1$. Then $L$ must contain two leaves from $\{\hat r_1,\hat r_2,\hat b\}$, and leaves in two of the sets $A_{\hat u}\setminus\{\hat u\}, A_{\hat v_1}\setminus\{\hat v_1\}, A_{\hat v_2}\setminus\{\hat v_2\}$ described above. Now, each of the sets $A_{\hat x}\setminus \{\hat x\}$ for $\hat x=\hat u,\hat v_1,\hat v_2$ are in the same tree in $T_2'$ at the end of the procedure, and at most one of these trees can contain $\lca_2(R\cup B)$. Hence, at most one of these three sets can contain leaves that are not in an inconsistent triplet with two leaves in $R\cup B$. Hence, a set that has its load increased by more than $1$ must be incompatible.
The dual solution therefore remains feasible.

Let $c=1$ if a {\it Merge-After-Cut} was performed in \resolvepair$(\hat v_1,\hat v_2)$ and $c=0$ otherwise, and let $d=1$ if line~\ref{alg:notalone} was executed to cut off $\hat u$ in $T_2'$, and $d=0$ otherwise.
We then have that $\Delta P_2=d+ 2 + (1-c)$: the ``2'' comes from the edge that was deleted from $T_2'$ to cut off $W=A_{\hat v_1}\setminus \{\hat v_1\}$ in \resolvepair$(\hat v_1,\hat v_2)$, and the edge that was deleted to cut off $\hat v_2$. 

We also have that $\Delta D_2 = d+2$: if $d=1$, then line~\ref{alg:notalone} increases the dual objective by $1$; \resolvepair$(\hat v_1,\hat v_2)$, when $\hat v_1,\hat v_2$ are in the same tree but not active siblings in $T_2'$, increases the dual objective value by $1$ (see Table~\ref{tab:resolvepair}); cutting of $\hat v_2$ from its tree in $T_2'$ increased the dual objective value by $1$.

We thus have $\Delta D_2=\tfrac12\Delta P_2+\tfrac 12(1+ c+d)$.
\end{itemize}

We now let $\Delta P_1$ be the number of edges in $E(T_2')\setminus E(\tilde T_2')$  due to lines~\ref{alg:beginwhile}--\ref{alg:endresolveB}, and we let $\Delta D_1$ be the total change in the dual objective value by lines~\ref{alg:beginwhile}--\ref{alg:endresolveB}. 
We let $c'=1$ if at least one {\it FinalCut} or {\it Merge-After-Cut} was executed in  lines~\ref{alg:beginwhile}--\ref{alg:endresolveB}. 
As in the proof of Lemma~\ref{lemma:threetrees}, we have that $\Delta D_1\ge \tfrac 12\Delta P_1-1 + \tfrac 12c'$.

From the discussion above, in the case when the procedure merges $\hat v_1$ and $\hat v_2$ as an active sibling pair in $T_2'$ in line~\ref{alg:rbmerge}, then
$\Delta D_1+\Delta D_2\ge \tfrac 12\Delta P_1-1 + \tfrac 12c' + 1+\Delta P_2 \ge \tfrac12\left(\Delta P_1+\Delta P_2\right)$.

Otherwise, we have
\[\Delta D_1+\Delta D_2\ge \tfrac 12\Delta P_1-1 + \tfrac 12c' + \tfrac 12\Delta P_2+\tfrac12(1+c+d)
\ge \tfrac12\left(\Delta P_1+\Delta P_2\right)-\tfrac12+\tfrac12(c'+c+d).\]

Therefore, if $c'+c+d\ge1$, we we have $\Delta D_1+\Delta D_2\ge\tfrac12\left(\Delta P_1+\Delta P_2\right)$. Furthermore, we note that if $c'+c+d=0$, then the condition in line~\ref{alg:ifretro2} is satisfied, and thus the retroactive merge ensures in this case that $|E(T'_2)\setminus E(\tilde T_2')|=\Delta P_1+\Delta P_2-1$, and we again have
$\Delta D_1+\Delta D_2\ge \tfrac 12|E(T'_2)\setminus E(\tilde T_2')|$. 
\end{proof}

%We have thus proved our main result:
Combining  Lemmas~\ref{lemma:success}, \ref{lemma:onetree}, \ref{lemma:threetrees} and~\ref{lemma:twotrees}, and noting that the dual objective value is a lower bound on the objective value of the optimal solution, we have thus proved our main result:

\setcounter{theorem}{2}
%\setcounterref{theorem}{thm:2approx}
\addtocounter{theorem}{-1}
\begin{theorem}
Algorithm~\ref{fig:2approx} is a 2-approximation algorithm for the Maximum Agreement Forest problem.
\end{theorem}

}

\section{Implementation Details}\label{sec:impl}
We implemented the Red-Blue approximation algorithm in Java, and tested it on instances with $|{\cal L}| = 2000$ leaves that were generated as follows: the number of leaves in the left subtree is set equal to a number between $1$ and $|{\cal L}|-1$ drawn uniformly at random, and a subset of this size is chosen uniformly at random from the label set. Then this procedure recurses until it arrives at a subtree with only $1$ leaf --- this will be the whole subtree.
 
After generating $T_1$ as described above, the tree $T_2$ was created by doing $50$ random Subtree Prune-and-Regraft operations (where random means that the root of the subtree that is pruned was chosen uniformly at random, as well as the edge which is split into two edges, so that the new node created can be the parent of the pruned subtree, under the conditions that this is a valid SPR-operation). This construction allows us to deduce an upper bound of 50 on the optimal value.  Our algorithm finds a dual solution that in $44\%$ of the 1000 runs is equal to the optimal dual solution, and in $37\%$ of the runs is $1$ less than the optimal solution. The observed average approximation ratio is about $1.92$.
After running our algorithm, we run a simple greedy search algorithm which repeatedly looks for two trees in the agreement forest that can be merged (i.e., 
such that the resulting forest is still a feasible solution to MAF).
The solution obtained after executing the greedy algorithm decreases the observed approximation ratio to less than $1.28$. 
The code is available at~\url{http://frans.us/MAF}.

\section{Conclusion}\label{sec:concl}
\iftoggle{abs}{}{We have shown how to construct an agreement forest for two rooted binary input trees $T_1$ and $T_2$ along with a feasible dual solution to a new LP relaxation for the problem. The objective value of the dual solution is at least half the number of components in the agreement forest.
Since the objective value of any dual solution gives a lower bound on the optimal value, this implies that our algorithm is a 2-approximation algorithm for MAF.
This improves on the previous best approximation guarantee of 2.5 by Shi et al.~\cite{Shietal15}.}

Our algorithm and analysis raise a number of questions. 
First of all, although we believe that, conceptually, our algorithm is quite natural, the actual algorithm is complicated, and it would be interesting to find a simpler 2-approximation algorithm.
Secondly, it is clear that our algorithm can be implemented in polynomial time, but the exact order of the running time is not clear. The bottleneck seems to be the finding of a minimal incompatible active sibling set, although it may be possible to implement the algorithm in a way that simultaneously processes sibling pairs as in \resolvepair, while it is looking for a minimal incompatible active sibling set.

\paragraph*{Acknowledgements}
We thank Neil Olver and Leen Stougie for fruitful discussions.

\bibliography{maf}{}

\begin{thebibliography}{10}

\bibitem{Allen2001}
Benjamin~L. Allen and Mike Steel.
\newblock Subtree transfer operations and their induced metrics on evolutionary
  trees.
\newblock {\em Annals of Combinatorics}, 5(1):1--15, 2001.

\bibitem{Bonet06}
Maria~Luisa Bonet, Katherine~St John, Ruchi Mahindru, and Nina Amenta.
\newblock Approximating subtree distances between phylogenies.
\newblock {\em Journal of Computational Biology}, 13(8):1419--1434, 2006.

\bibitem{Bordewich08}
Magnus Bordewich, Catherine McCartin, and Charles Semple.
\newblock A 3-approximation algorithm for the subtree distance between
  phylogenies.
\newblock {\em Journal of Discrete Algorithms}, 6(3):458--471, 2008.

\bibitem{Bordewich04}
Magnus Bordewich and Charles Semple.
\newblock On the computational complexity of the rooted subtree prune and
  regraft distance.
\newblock {\em Ann. Comb.}, 8(4):409--423, 2004.

\bibitem{Darwin37}
Charles Darwin.
\newblock {\em Notebook B: Transmutation of species (1837--1838)}.
\newblock In: John van Wyhe: The Complete Work of Charles Darwin Online, 2002.
\newblock \url{http://darwin-online.org.uk/}.

\bibitem{Farach94}
Martin Farach and Mikkel Thorup.
\newblock Optimal evolutionary tree comparison by sparse dynamic programming.
\newblock In {\em FOCS '94: Proceedings of 35th Annual Symposium on Foundations
  of Computer Science}, pages 770--779. IEEE, 1994.

\bibitem{Farach97}
Martin Farach and Mikkel Thorup.
\newblock Sparse dynamic programming for evolutionary-tree comparison.
\newblock {\em SIAM Journal on Computing}, 26(1):210--230, 1997.

\bibitem{Gordon79}
{A.\,D.} Gordon.
\newblock A measure of the agreement between rankings.
\newblock {\em Biometrika}, 66(1):7--15, 1979.

\bibitem{Hein96}
Jotun Hein, Tao Jiang, Lusheng Wang, and Kaizhong Zhang.
\newblock On the complexity of comparing evolutionary trees.
\newblock {\em Discrete Applied Mathematics}, 71(1-3):153--169, 1996.

\bibitem{Rodrigues03}
Estela~M. Rodrigues.
\newblock {\em Algoritmos para Compara\c c\~ao de \'Arvores Filogen\'eticas e o
  Problema dos Pontos de Recombina\c c\~ao}.
\newblock PhD thesis, University of S\~ao Paulo, Brazil, 2003.
\newblock Chapter 7, available at
  \url{http://www.ime.usp.br/~estela/studies/tese-traducao-cp7.ps.gz}.

\bibitem{Rodrigues07}
Estela~M. Rodrigues, Marie-France Sagot, and Yoshiko Wakabayashi.
\newblock The maximum agreement forest problem: approximation algorithms and
  computational experiments.
\newblock {\em Theoretical Computer Science}, 374(1-3):91--110, 2007.

\bibitem{Shietal15}
Feng Shi, Qilong Feng, Jie You, and Jianxin Wang.
\newblock Improved approximation algorithm for maximum agreement forest of two
  rooted binary phylogenetic trees.
\newblock {\em Journal of Combinatorial Optimization}, 2015.

\bibitem{Steel93}
Mike Steel and Tandy Warnow.
\newblock Kaikoura tree theorems: Computing the maximum agreement subtree.
\newblock {\em Information Processing Letters}, 48(2):77--82, November 1993.

\bibitem{Whidden13}
Chris Whidden, Robert~G. Beiko, and Norbert Zeh.
\newblock Fixed-parameter algorithms for maximum agreement forests.
\newblock {\em SIAM Journal on Computing}, 42(4):1431--1466, 2013.

\bibitem{Whidden09}
Chris Whidden and Norbert Zeh.
\newblock A unifying view on approximation and {FPT} of agreement forests.
\newblock In {\em Algorithms in Bioinformatics}, volume 5724 of {\em Lecture
  Notes in Computer Science}, pages 390--402. Springer Berlin Heidelberg, 2009.

\bibitem{Wu09}
Yufeng Wu.
\newblock A practical method for exact computation of subtree prune and regraft
  distance.
\newblock {\em Bioinformatics}, 25(2):190--196, 2009.

\bibitem{Wu10}
Yufeng Wu and Jiayin Wang.
\newblock Fast computation of the exact hybridization number of two
  phylogenetic trees.
\newblock In {\em Bioinformatics Research and Applications}, volume 6053 of
  {\em Lecture Notes in Computer Science}, pages 203--214. Springer Berlin
  Heidelberg, 2010.

\end{thebibliography}
\iftoggle{abs}{}{\appendix\section{Deferred Proofs}\label{sec:defer}
%\subsection{Proofs for Section~\ref{sec:multitrees}}
\setcounter{claim}{2}
\addtocounter{claim}{-1}
%\setcounterref{claim}{realcut}
%\addtocounter{claim}{-1}

\begin{claim}
Just before line~\ref{alg:cutb} and line~\ref{alg:cutp}, the leaf $\hat b$ and $\hat v_2$, respectively, is not the only active leaf in its tree in $T_2'$. 
\end{claim}

\begin{proofof}{Claim~\ref{realcut}}
We consider $T_2'$ in line~\ref{alg:triplet}, and we note that $\hat r_1,\hat r_2$ and $\hat b$ are contained in multiple trees in $T_2'$.
We show that, if $\hat r_1,\hat r_2$ are in the same tree in $T_2'$, then that tree contains at least one active leaf $w$ that is not a descendent of $\lca_2(R\cup B)$. Otherwise, we show that the tree in $T_2'$ containing $\hat b$ contains at least one active leaf $w$ that is not a descendent of $\lca_2(R\cup B)$.
This suffices to prove the claim, since the leaf $w$ will be in the same tree as $\hat b$ if line~\ref{alg:cutb} is executed, or in the same tree as $\hat v_2$ if line~\ref{alg:cutp} is executed.

First, suppose $\hat r_1$ and $\hat r_2$ are in the same tree in $T_2'$. Then this tree contains $\lca_2(R\cup B)=\lca_2(\hat r_1,\hat r_2)$. It then follows from Observation~\ref{obs:fail} (3) and the fact that lines~\ref{alg:beginresolveB}--\ref{alg:endresolveB} cannot delete an edge in $T_2'$ above $\lca_2(R\cup B)$ that the tree containing $\hat r_1,\hat r_2$ contains a leaf  $w$ that is not a descendent of $\lca_2(R\cup B)$.

Otherwise, suppose $\hat r_1$ and $\hat r_2$ are in different trees in $T_2'$. We will show that this implies that $\hat b$ must be in the same tree as $\lca_2(R\cup B)$, and thus that the tree containing $\hat b$ in $T_2'$ contains at least one leaf $w$ that is not a descendent of $\lca_2(R\cup B)$, by the same reasonining as above. 
First of all, note that if the tree containing $\lca_2(R\cup B)$ contains at least one active leaf in $B$ at the start of line~\ref{alg:beginresolveB}, then the same holds after line~\ref{alg:endresolveB}: this is because otherwise there must be a {\it FinalCut} that cuts off the last leaf, say $u$, in $B$ from this tree in a call to \resolvepair$(u,v)$, with $v\in B$. But since $\lca_2(u,v)$ is on the path from $u$ to $\lca_2(R\cup B)$, the tree containing $u$ and $\lca_2(R\cup B)$ contains $\lca_2(u,v)$, and thus we would have relabeled $u$ and $v$ and cut off the other leaf in the pair, by line~\ref{alg:cond1} of \resolvepair. %REFER TO PROPERTIES
Thus $\hat b$ must be in the same tree as $\lca_2(R\cup B)$, unless this tree contained no active leaves in $B$ at the start of line~\ref{alg:beginresolveB}. But in the latter case, no edges are cut from the tree containing $\lca_2(R\cup B)$ by lines~\ref{alg:beginresolveB}--\ref{alg:endresolveB}, and thus, using Observation~\ref{obs:fail} (2), $\hat r_1$ and $\hat r_2$ would still be in the same tree in $T_2'$, contradicting our assumption.
\end{proofof}

%\setcounterref{claim}{claim:retro}
%\addtocounter{claim}{-1}

\begin{claim}
The leaves $u',v'$ in line~\ref{alg:retromerge1} and line~\ref{alg:retromerge2} exist. Furthermore, the set $W$ in line~\ref{alg:retromerge1} and line~\ref{alg:retromerge2} contains only active leaves that have the same color as $\hat r_2$ and $\hat u$, respectively.
\end{claim}

\begin{proofof}{Claim~\ref{claim:retro}}
To simplify notation, we will let $x$ denote $\hat r_2$ or $\hat u$ respectively. Let $X$ be the color of $x$, i.e., $X$ is $R$ or $B$.
We say an active tree in $T_2'$ is $X$-bicolored tree if it contains active leaves in $X$ and active leaves that are not in $X$, and we will say it is an $X$-unicolored tree if all its active leaves are in $X$.
Consider the first moment when the tree containing $x$ in $T_2'$ is $X$-unicolored, i.e., when all active leaves in the tree in $T_2'$ containing $x$ are in $X$. 
There are three options: (1) this is already the case at the point where the current sets $R$ and $B$ were defined, i.e., in line~\ref{alg:beginwhile}, 
(2) this is achieved through a call to \resolvepair\ for an active sibling pair $\tilde u, \tilde v \in X$, and (3) this is achieved through a call to \resolvepair\ for an active sibling pair $\tilde u, \tilde v \not\in X$.

We will show that if (1) or (2) happens, then Algorithm~\ref{fig:2approx} must have executed at least one {\it FinalCut} or {\it Merge-After-Cut} in some call to \resolvepair\ in lines~\ref{alg:beginwhile}--\ref{alg:endresolveB}. Since this contradicts the condition in line~\ref{alg:ifretro1}, respectively line~\ref{alg:ifretro2}, it must be the case that (3) happens, and this is exactly what is stated in the claim.

For (2), note that a call to \resolvepair$(\tilde u,\tilde v)$ where $\tilde u,\tilde v\in X$ are an active sibling pair in $T_1'$ can only change a $X$-bicolored tree into an $X$-unicolored tree if it cuts off a set of leaves that are not in $X$ and what remains is a tree containing only active leaves in $X$, i.e., either line~\ref{alg:finalcut}  in \resolvepair, or line~\ref{alg:2cutW} followed by line~\ref{alg:cutmerge}. Hence, in case (2) the call to \resolvepair$(\tilde u, \tilde v)$ executes a {\it FinalCut} or a {\it Merge-After-Cut}.

For (1), note that the tree containing $x$ in $T_2'$ in line~\ref{alg:beginwhile} must contain at least one other active leaf in $X$, say $x'$, since $x$ was not deactivated by the preprocessing. Since we also did not merge these leaves in \preprocess, they must not have been active siblings in $T_1'$. Hence, there must have been some other tree in $T_2'$ in line~\ref{alg:beginwhile} that contains $x''\in X$ such that $x'x''|x$ in $T_1$ or $xx''|x'$ in $T_1$. Repeating this argument if necessary shows that there exists an $X$-bicolored tree in $T_2'$ at the start of the while-loop. Furthermore, the argument also shows that there exists such an $X$-bicolored tree, say with active leaf set $A$, such that $\lca_2(A)$ is a strict descendent of $\lca_2(x,x')$, and hence also a strict descendent of $\lca_2(R\cup B)$, in the original tree $T_2$.

We now show that there must have been at least one {\it FinalCut} or {\it Merge-After-Cut} executed on the leaves in $A\cap (R\cup B)$.

First, suppose $A$ contains leaves in $R$, i.e., the tree in $T_2'$ containing $A$ at the start of the while-loop is $R$-bicolored. Note that $\hat r_1\not\in A, \hat r_2\not\in A$: when \resolveset$(R)$ fails, $\hat r_1$ and $\hat r_2$ are in the same tree in $T_2'$, and hence, this must have also been true at the start of the while-loop. But since $\lca_2(\hat r_1,\hat r_2)=\lca_2(R\cup B)$, this means that $\hat r_1, \hat r_2$ cannot both be in $A$, as this contradicts the fact that $\lca_2(A)$ is a strict descendent of $\lca_2(R\cup B)$. 
But if $A$ does not contain $\hat r_1,\hat r_2$, then all leaves in $A\cap R$ will be inactive after executing \resolveset$(R)$. Hence, by Observation~\ref{obs:bicolor}, at least one {\it FinalCut} or {\it Merge-After-Cut} must have been performed during \resolveset$(R)$.  

Now, suppose $A$ contains no leaves in $R$. Then $X$ must be $B$. The execution of \resolveset$(R)$ does not affect this tree, and since the leaves $x,x'$ are merged at some point during lines~\ref{alg:beginresolveB}--\ref{alg:endresolveB}, the leaves in $A\cap B$ must all be deactivated by lines~\ref{alg:beginresolveB}--\ref{alg:endresolveB}.  Hence, we can again invoke Observation~\ref{obs:bicolor}, at least one {\it FinalCut} or {\it Merge-After-Cut} must have been performed in lines~\ref{alg:beginresolveB}--\ref{alg:endresolveB}.
\end{proofof}

%\setcounterref{claim}{claim:freenodes}
%\addtocounter{claim}{-1}

\begin{claim}
Let $u'$ be the leaf and $W$ the set identified in line~\ref{alg:retromerge1} or line~\ref{alg:retromerge2}, respectively, and let $L_1, \ldots, L_k$ be the leaf sets of the trees in $T_1'$ at this moment, excluding the trees containing $u'$ and $W$, and let $M_1, \ldots, M_{k'}$  be the leaf sets of the trees in $T_2'$ at this moment, excluding the trees containing $u'$ and $W$. 
Then, $V_1[L_j]\cap V_1[\{u'\}\cup W]=\emptyset$ for all $j=1,\ldots, k$ and $V_2[M_{j'}]\cap V_2[\{u'\}\cup W]=\emptyset$ for all $j'=1,\ldots, k'$.
\end{claim}

\begin{proofof}{Claim~\ref{claim:freenodes}}
We first consider $T_2'$. Let $v'$ be such that the execution of \resolvepair$(u',v')$ cut off set $W$ in $T_2'$. Since no {\it Merge-After-Cut} was executed, we know that $u'$ was also cut off by \resolvepair$(u',v')$.
Note that before this execution of \resolvepair$(u',v')$, $p_2(u')=p_2(W)$; we will refer to this node as $p$. The edge below $p$ was deleted to cut off $W$ in line~\ref{alg:2cutW}. Then, in line~\ref{alg:cutu} the edge below the new parent of $u'$ was deleted. Therefore, $p$ is in the tree in $T_2'$ containing $u'$, and  thus all nodes in $T_2$ in $V[\{u'\}\cup W]$ are either in the tree in $T_2'$ containing $W$ or the tree containing $u'$. After \resolvepair$(u',v')$, the tree containing $u'$ in $T_2'$ is inactive, and for the tree in $T_2'$ containing $W$, the active leaves are either all in $R$ or all in $B$ by Claim~\ref{claim:retro}. Therefore, the executions of \resolvepair\ until we reach line~\ref{alg:retromerge1} or line~\ref{alg:retromerge2}, respectively, do not affect the edge set of the tree in $T_2'$ containing $W$.  Hence for the nodes in $T'_2$ we can conclude that $V_2[M_{j'}]\cap V_2[\{u'\}\cup W]=\emptyset$ for all $j'=1,\ldots, k'$.

Now consider $T_1'$. Suppose, by means of contradiction that $V_1[L_{j}]\cap V_1[\{u'\}\cup W] \neq \emptyset$ for some $L_j$. 
We have the following claim.

%\setcounterref{claim}{claim:rbmerge}
\setcounter{claim}{6}
\addtocounter{claim}{-1}
\begin{claim}\label{claim:xy}
If $V_1[L_{j}]\cap V_1[\{u'\}\cup W] \neq \emptyset$, then there exist $x,y\in L_j$ that are active at the moment that \resolvepair$(u',v')$ was executed, such that
$V_1[\{x,y\}]\cap  V_1[\{u'\}\cup W]\neq \emptyset$ and there was some execution of \resolvepair$(x,y)$ after \resolvepair$(u',v')$ in which $x$ and $y$ are merged.
\end{claim}

First, we note that all leaves in $L_j$ must be inactive, since all nodes in $V_1[\{u'\}\cup W]$ are descendents of $\lca_1(R\cup B)$, and all leaves that are descendents of $\lca_1(R\cup B)$ in $T_1'$ are in inactive trees.
The moment that the last leaf in $L_j$ was made inactive, we cut off the {\it subtree} in $T_1'$ represented by this leaf (i.e., the leaves that were (recursively) merged with it).  

We now consider the active leaves excluding $W\cup\{u'\}$ just before the execution of \resolvepair$(u',v')$ and the subtrees in $T_1'$ that they represent.  Note that these subtrees are not components of $T_1'$, but that, for a given active leaf $u$, the subtree represented by $u$ is one of the two subtrees rooted at the children of $p_1(u)$, namely the subtree that contains $u$. 
By definition, these subtrees contain exactly one active leaf, and therefore these subtrees do not contain a leaf in $W\cup \{u'\}$.
It follows that none of the nodes in $V_1[\{u'\}\cup W]$, are in such a subtree: Since all active leaves are in the same tree in $T_1'$, and since the leaves in $\{u'\}\cup W$ are active, a subtree containing a node in $V_1[\{u'\}\cup W]$ would also contain either $u'$ or $W$.

Hence, there must have been some moment, between the execution of \resolvepair$(u',v')$ and the moment of the retroactive merge, where some node in $V_1[\{u'\}\cup W]$ became part of the subtree represented by some active leaf, say $y$, and the only way in which the subtree represented by $y$ can change is when some active leaf $x$ is merged with $y$. Furthermore, the node in $V_1[\{u'\}\cup W]$ that becomes part of the subtree represented by $y$ must be in $V_1[\{x,y\}]$. We have thus proven Claim~\ref{claim:xy}.
\vspace*{.001\baselineskip}

We now show how Claim~\ref{claim:xy} can be used to prove Claim~\ref{claim:freenodes}.
First, suppose $x,y$ are a different color from $u',v'$, i.e., if $u',v'\in R$, then $x,y\in B$ and vice versa. Note that then $W$ has the same color as $x,y$, by Claim~\ref{claim:retro}. Since $V_1[\{x,y\}]$ intersects $V_1[\{u'\}\cup W]$, and $u'$ has a different color from $W$, we have that $W$ is a descendent of $\lca_1(x,y)$. But note that one leaf in $W$, either $\hat r_2$ or $\hat u$, is still active at the start of Procedure~\ref{fig:threetrees} and Procedure~\ref{fig:twotrees}. Hence, $x$ and $y$ cannot have been merged before the start of this procedure. Furthermore, if there was a merge in Procedure~\ref{fig:twotrees} (in line~\ref{alg:rbmerge} or line~\ref{alg:lastpair}) we would not have executed a retroactive merge, i.e., we would not have reached line~\ref{alg:retromerge2}. 

We now consider the case that $x,y$ have the same color as $u',v'$. Let $U$ be the color of $x,y,u',v'$, i.e., $U=R$ if $x,y,u',v'\in R$ and $U=B$ otherwise. We let $\tilde U$ be the leaves in $U$ that are active just before executing \resolvepair$(u',v')$ that are descendents of $\lca_1(x,y)$.
We will use Observation~\ref{obs:bicolor} to show that there must have been a {\it FinalCut} or a {\it Merge-After-Cut} executed in a call to \resolvepair\ with arguments $\tilde u_1,\tilde u_2\in \tilde U$, between time $t_1$ (just before the call to \resolvepair$(u',v')$) and $t_2$ (just after the call to \resolvepair$(x,y)$ in which $x$ and $y$ are merged). 

Since $\tilde U$ contains the active leaves in the subtree of $T_1'$ rooted at $\lca_1(x,y)$ which is a subtree of the subtree of $T_1'$ rooted at $\lca_1(U)$, 
executions of \resolvepair\ between time $t_1$ and $t_2$ are performed for active sibling pairs in $T_1'$ that are either both in $\tilde U$, or both in $U \setminus \tilde U$. Furthermore, the latter do not affect any part of $T_2'$ that is not below $\lca_2(\tilde U)=\lca_2(x,y)$, and we may therefore assume that all calls to \resolvepair\ between time $t_1$ and $t_2$ are performed for active sibling pairs $u,v\in \tilde U$.
In fact, we may assume for simplicity that $\lca_2(x,y)$ is the root of the tree in $T_2'$ containing $x$ and $y$, since this does not affect the outcome of any call to \resolvepair\ between time $t_1$ and $t_2$.

We let $A(x,y)$ be the leaves that are descendents of $\lca_2(x,y)$ in the original tree $T_2$, and that are active at time $t_1$. 
Note that $u'\in A(x,y)$, since $u'$ must be a descendent of $\lca_1(x,y)$ by the condition that $V_1[\{x,y\}]\cap V_1[W\cup \{u'\}]\neq\emptyset$, and by the fact that $u',x,y\in U$ and $U$ is compatible, $u'$ must then also be a descendent of $\lca_2(x,y)$.
Also, note that $u'$ and $v'$ were active siblings in $T_1'$ when executing \resolvepair$(u',v')$, so also $v'\in A(x,y)$. Finally, this implies that also $W\subset A(x,y)$, since $W$ is the set of active leaves that is cut off in line~\ref{alg:2cutW} of \resolvepair$(u',v')$. Note that $W,u',v'$ are all in the same tree in $T_2'$ at time $t_1$, and hence at this moment, there exists a tree in $T_2'$ containing active leaves $A\subseteq A(x,y)$ that is $\tilde U$-bicolored. 
At time $t_2$, we have just merged $x$ with $y$, and so the only leaf in $\tilde U$ that is still active is $y$.
Furthermore, by our simplifying assumption that $\lca_2(x,y)$ is the root of the tree in $T_2'$ containing $x$ and $y$, the tree in $T_2'$ containing $y$ has no other active leaves, and is therefore $\tilde U$-unicolored. 
We have thus shown that at time $t_1$, there is a tree in $T_2'$ with active leaf set $A$ that is $\tilde U$-bicolored, and at time $t_2$, all leaves in $\tilde U$ (and hence in $A\cap \tilde U$) are inactive or in $\tilde U$-unicolored trees. Furthermore, we argued that we may assume that all calls to \resolvepair\ between time $t_1$ and $t_2$ are performed for active sibling pairs $u,v\in \tilde U$.

By Observation~\ref{obs:bicolor}, there must therefore have been some call to \resolvepair\ that performed a {\it FinalCut} or {\it Merge-After-Cut} between time $t_1$ and $t_2$.
\end{proofof}

\setcounter{claim}{5}
\addtocounter{claim}{-1}
%\setcounterref{claim}{claim:rbmerge}
%\addtocounter{claim}{-1}

\begin{claim}
If $\hat v_1$ and $\hat v_2$ are active siblings in $T_2'$ in line~\ref{alg:ifsib} of Procedure~\ref{fig:twotrees}, then one of them is $\hat b$.
\end{claim}
\begin{proofof}{Claim~\ref{claim:rbmerge}}
Suppose by contradiction that $\hat v_1$ and $\hat v_2$ are $\hat r_1$ and $\hat r_2$.
First, suppose the tree containing $\hat r_1,\hat r_2$ contained some active leaf in $B$ after \resolveset$(R)$.  Since, at the start of Procedure~\ref{fig:twotrees}, $\hat b$ is the only active leaf in $B$, and $\hat b$ is not in the same tree as $\hat r_1,\hat r_2$, it must be the case that some execution of \resolvepair$(u,v)$ with $u,v\in B$ cut off the leaf $u$, and after this operation the tree in $T_2'$ containing $\hat r_1,\hat r_2$ contains no active leaves in $B$. But note that, since $u,v\in B$,  the path from $u$ to $\lca_2(R\cup B)$ must contain $\lca_2(u,v)$, and thus, since $\lca_2(\hat r_1,\hat r_2)=\lca_2(R\cup B)$, the tree containing $u,\hat r_1$ and $\hat r_2$ before \resolvepair$(u,v)$ contained $\lca_2(u,v)$. But this contradicts the conditon in line~\ref{alg:cond1} of \resolvepair\ for defining $u$ if $u$ and $v$ are in different trees in $T_2'$.

Hence, it must be the case that, if $\hat v_1,\hat v_2$ are $\hat r_1,\hat r_2$, then the tree in $T_2'$ containing $\hat r_1,\hat r_2$ after \resolveset$(R)$ contained no leaves in $B$.
Note that this immediately gives a contradiction, since then the tree containing $\hat r_1,\hat r_2$ was not changed after \resolveset$(R)$, and thus $\hat r_1,\hat r_2$ must have been active siblings at the moment when \resolveset$(R)$ failed, but this contradicts Observation~\ref{obs:fail}~(4).
Thus, $\hat v_1$ and $\hat v_2$ cannot be $\hat r_1$ and $\hat r_2$, so one of them must be $\hat b$.
\end{proofof}}
\end{document}